\documentclass[twoside]{article}

\newcommand{\mytitle}{A Solvable Mixed Charge Ensemble on the Line:  Global Results}

\usepackage{amsfonts, amssymb, amsmath, amsthm}
\usepackage[dvips]{graphicx}

\newtheorem{thm}{Theorem}[section]
\newtheorem{cor}[thm]{Corollary}
\newtheorem{lemma}[thm]{Lemma}
\theoremstyle{remark}
\newtheorem*{rem}{Remark}

\newcommand{\G}[1]{\Gamma\left( #1 \right)}
\newcommand{\transpose}{\ensuremath{\mathsf{T}}}
\newcommand{\ul}[1]{\underline{#1}}
\newcommand{\wt}{\ensuremath{\widetilde}}
\newcommand{\la}{\ensuremath{\langle}}
\newcommand{\ra}{\ensuremath{\rangle}}
\newcommand{\Prob}{\ensuremath{\mathrm{Prob}}}
\newcommand{\Var}{\ensuremath{\mathrm{Var}}}
\newcommand{\bs}{\ensuremath{\boldsymbol}}
\newcommand{\mf}{\ensuremath{\mathfrak}}
\newcommand{\BB}[1]{\ensuremath{\mathbb{#1}}}
\newcommand{\R}{\ensuremath{\BB{R}}}
\newcommand{\E}{\ensuremath{\BB{E}}}

\newcommand{\wh}{\widehat}
 \def\a{{\alpha}} 
 \def\RR{{\mathbb R}}

\DeclareMathOperator{\Pf}{Pf}
\DeclareMathOperator{\sgn}{sgn}

\numberwithin{equation}{section} 

\numberwithin{equation}{section}

\begin{document}
\title{\bf \mytitle}  
\author{{\sc Brian Rider}\footnote{partially supported
by NSF grant DMS-0645756.}, {\sc Christopher
D.~Sinclair}\footnote{partially supported by NSF grant DMS-0801243},
{\sc Yuan Xu}}
\maketitle

\begin{abstract}
  We consider an ensemble of interacting charged particles on the line
  consisting of two species of particles with charge ratio $2:1$ in
  the presence of the harmonic oscillator potential.  The system is
  assumed to be at temperature corresponding to $\beta=1$ and the
  sum of the charges is fixed.  We investigate the
  distribution of the number as well as the spatial
  density of each species of particle in the limit as the total charge
  increases to $\infty$.  These results will follow from the fact that
  the system of particles forms a Pfaffian point process.  We produce
  the skew-orthogonal polynomials necessary to simplify the related
  matrix kernels.
\end{abstract}

\section{Introduction}

Let $L,M$ and $N$ be non-negative integers so that $L + 2M = N$, and
consider 1-dimensional electrostatic system consisting of $L$
particles with unit charge and $M$ particles with charge 2.  We will
identify the state of the system by pairs of finite subsets of $\R$,
$\xi_1 = \{\alpha_1, \alpha_2, \ldots, \alpha_L\}$ and $\xi_2 = \{
\beta_1, \beta_2, \ldots \beta_M\}$, where $\alpha_1, \alpha_2, \ldots,
\alpha_L$ represent the locations  of the charge 1 particles and
$\beta_1, \beta_2, \ldots, \beta_M$ represent the locations of the
charge 2 particles.   

The potential energy of state $\xi = (\xi_1, \xi_2)$ is given by 
\[
\sum_{j<k} \log| \alpha_j - \alpha_k| +
4 \sum_{m < n} \log| \beta_m - \beta_n | + 2 \sum_{\ell=1}^L
\sum_{m=1}^M |\alpha_{\ell} - \beta_m |.
\] 
We assume that the system is in the presence of an external field, so
that the interaction energy between the charges and the field is given
by
\[
-\sum_{\ell=1}^L V(\alpha_{\ell}) -2 \sum_{m=1}^M V(\beta_m)
\]
for some potential $V: \R \rightarrow [0,\infty)$.  Eventually we will
specify to the situation where $V$ is the harmonic
oscillator potential, but for now we maintain generality.  The total
potential energy of the system is therefore
\begin{align}
&E = \sum_{j<k} \log| \alpha_j - \alpha_k| +
4 \sum_{m < n} \log| \beta_m - \beta_n | + 2 \sum_{\ell=1}^L
\sum_{m=1}^M \log |\alpha_{\ell} - \beta_m | \label{eq:16} \\ & \hspace{7cm}
-\sum_{\ell=1}^L V(\alpha_{\ell}) -2 \sum_{m=1}^M V(\beta_m). \nonumber
\end{align}
Given a pair of vectors $(\bs \alpha, \bs \beta) \in \R^L \times \R^M$
we will define $E(\bs \alpha, \bs \beta)$ to be the right hand side of
(\ref{eq:16}), and call $(\bs \alpha, \bs \beta)$ a state {\em vector}
corresponding to the state $\xi$.  Generically, there are $L! M!$
state vectors corresponding to a given state.  

Assuming the system is placed in a heat bath corresponding to inverse
temperature parameter $\beta=1$, then the Boltzmann factor for the
state vector $(\bs \alpha, \bs \beta)$ is given by
\begin{equation}
\label{eq:2}
e^{-E(\bs \alpha, \bs \beta) } = \prod_{\ell=1}^L
w(\alpha_{\ell}) \prod_{m=1}^M w(\beta_m)^2 \prod_{j < k} | \alpha_j -
\alpha_k | \prod_{m < n} |\beta_m - \beta_n|^4 \prod_{\ell=1}^L
\prod_{m=1}^M |\alpha_{\ell} - \beta_m|^2,
\end{equation}
where $w(\gamma) = e^{-V(\gamma)}$ is the {\em weight} of the system.  
The partition function of the system is given by 
\begin{equation}
\label{eq:2.5}
 Z_{L,M} = \frac{1}{L! M!} \int_{\R^L} \int_{\R^M} e^{ -E(\bs
  \alpha, \bs \beta) } d\mu^L(\alpha) \, d\mu^M(\beta),
\end{equation}
where $\mu$ and $\mu^L$ are Lebesgue measure on $\R$ and $\R^L$
respectively.  The multiplicative prefactor $1/(L! M!)$ compensates
for the multitude of state vectors associated to each state.  

Here we will be interested in a form of the grand canonical ensemble
conditioned so that the sum of the charges equals $N$.  That is, we
consider the union of all  
two component ensembles with $L$ particles of charge 1 and $M$
particles of charge 2 over all pairs of non-negative integers $L$ and
$M$ with $L + 2M = N$.  The partition function of this ensemble is given by  
\[
Z(X) = \sum_{(L,M)} X^{L} Z_{L,M} = \sum_{(L,M)} \frac{X^{L}}{L! M!} \int_{\R^L}
\int_{\R^M} e^{ -E(\bs \alpha, \bs \beta) } d\mu^L(\alpha) \,
d\mu^M(\beta).
\]
Here $X \ge 0$ is the {\em fugacity} of the system, a parameter which
controls the probability that the system has a particular population
vector $(L,M)$.  The sum over $(L,M)$ indicates that we are summing over all
pairs of non-negative integers such that $L + 2M = N$.

Note now that $(L,M)$ is itself a random vector, though we will 
continue to use this notation for the value of the population vector
as well.  For example,
for each admissible pair $(L,M)$, the joint density of particles given
population vector $(L,M)$ is given by the normalized Boltzmann factor,
\begin{equation}
\label{eq:3}
\frac{X^{L}}{Z(X)} e^{-E(\bs \alpha, \bs \beta)}.
\end{equation}
When $X=1$ the probability of seeing a particular pair $(L,M)$, or
$\Prob(L,M)$, is the ratio $Z_{L,M}/Z$, where $Z= Z(1)$.   

Experts of random matrix theory will have already noticed that when $X=0$ the above
reduces to a general orthogonal (or $\beta=1$) ensemble.  Likewise, as $X \rightarrow \infty$,
the above formally goes over to the corresponding symplectic (or $\beta =4$) ensemble.    
This provides then an unusual sort of interpolation between two classical and well studied point
processes.

\section{Statement of results}

In this paper we will primarily be concerned with global statistics of
the particles when the fugacity equals 1 and the potential $V$ is
given by
\[
V(\gamma) = \gamma^2/2, \qquad \mbox{that is,} \qquad w(\gamma) =
e^{-\gamma^2/2}.
\]
Many of the results presented here are valid for other potentials and
other values of $X$, however unless otherwise indicated we will
restrict ourselves to these choices of $V$ and $X$.  We will also
restrict ourselves to the situation where $N = 2J$ is an even
integer.  

Similar results for the two-charge ensemble constrained to the circle
with uniform weight were obtained by P.J.\ Forrester (see 5.9 of \cite{forrester-book}
and the references therein).  

The goal of this paper is to present global results about the
distribution of $L$ and $M$ as well as the global spatial distribution
of each of the species of particles.  Along the way we will derive a
Pfaffian point process for the particles (similar to that of  another
two-component ensemble, Ginibre's real ensemble) as well as the
skew-orthogonal polynomials which allow us to present a simplified
matrix kernel for the process.  The local analysis of
this kernel ({\it i.e.}~its scaling limits in the bulk and at the
edge) as well an investigation of the right-most particle of each species
will appear in a forthcoming publication.  

\subsection{Distribution of the population vectors}

Sharp results on the law of the state
vector $(L,M)$ are consequences of the following characterization. 

\begin{thm}
\label{thm:6}
For each non-negative integer $j$, let
$L_j = L_j^{(-1/2)}$ be the $j$th Laguerre polynomial with
parameter $\alpha = -1/2$.  Then, $\Prob(L,M)$ is the coefficient of
$X^L$ of the polynomial $L_{N/2}(-X^2)/L_{N/2}(-1)$.  That is, 
\[
\frac{Z(X)}Z = \frac{L_{N/2}(-X^2)}{L_{N/2}(-1)},
\]
and so
\begin{enumerate}
\item ${\displaystyle 
\Prob(L,M) = \frac{2^L}{L! M!} \sum_{(\ell,m) \atop \ell + 2m = N}
\frac{2^{\ell}}{\ell! m!} \qquad \mbox{if $L$ is even, and is equal to
  0 otherwise,}}$
\item  ${\displaystyle 
\E (L^m)  =   \left[ ( X \frac{d}{dX} )^m  \frac{L_{N/2}(-X^2)}{L_{N/2}(-1) } \right]_{X=1}  \qquad \mbox{for $m$ non-negative integer}.}$
\end{enumerate}
\end{thm}

Properties of the Laguerre polynomials now  allow for nice expressions for the mean, variance, {\em{etc.}} of $L$ for all 
finite values of $N$.  For example,
we have that
$$
\E[L]  = \frac{d}{dX}\left[
      \frac{Z(X)}{Z} \right]_{X=1} =  2 \sum_{j=0}^{N/2-1}
    \frac{L_j(-1)}{L_{N/2}(-1)} = \frac{{\displaystyle 
2 \sum_{i=0}^{N/2-1} \left[ \G{\frac{N}2 - i} \G{i + \frac32} i! \right]^{-1}
}}{{\displaystyle
\sum_{i=0}^{N/2} \left[ \G{\frac{N}2 - i + 1} \G{i + \frac12} i! \right]^{-1}
}}.
$$
Asymptotic descriptions of the law of $L$ are just as readily obtained from Theorem~\ref{thm:6}.

\begin{thm} 
\label{thm:4}
As $N \rightarrow \infty$ it holds:

\begin{enumerate}
\item $\mathbb{E} (L) = \sqrt{2N}  - {1} + \frac{1}{3 \sqrt{N}} + \mathrm{O}(N^{-1})$ and 
$\Var (L) =   \sqrt{ 2 {N}}  - \frac{4}{3} + \mathrm{O}(N^{-1/2})$, 
\item $ {\displaystyle \frac{ L - (2N)^{1/2}}{ (2N)^{1/4}}} $ converges in
distribution to a standard Normal random variable,
\item  $ \Prob \left(  | \frac{L}{\sqrt{2N}} - 1|  \ge \epsilon  \right)  \le C N e^{- (\epsilon \wedge 1) \sqrt{2N}} $  with a numerical constant $C$ for
any $\epsilon > 0$.
\end{enumerate}
\end{thm}

\subsection{Spatial density of particles}
\label{sec:spat-dens-part}

We introduce the (mean) counting measures $\rho_1$ and $\rho_2$ for the
charge 1 and charge 2 particles defined by 
\[
\E\left[ |A \cap \xi_1| \right] = \int_A d\rho_1
\qquad \mbox{and} \qquad 
\E\left[ |A \cap \xi_2| \right] = \int_A d\rho_2
\]
for Borel subsets $A \subseteq \R$ (where, for instance, $|A \cap \xi_1|$ is
the number of charge 1 particles in $A$). As we shall see in the sequel,
these measures are absolutely continuous with respect to Lebesgue
measure, and we will write $R^{(N)}_{1,0}(x)$ and $R^{(N)}_{0,1}(x)$ for
their respective densities.  (The cryptic notation will be resolved
in Section~\ref{sec:corr-funct}, when we define the
$\ell,m$-correlation function of the ensemble to be
$R^{(N)}_{\ell,m}$).  

From Theorem~\ref{thm:4} we see that, 
as $N \rightarrow \infty$, 
\[
\int_{-\infty}^{\infty} R^{(N)}_{1,0}(x) \, dx = \E[L] \sim \sqrt{2N},
\quad \mbox{ and } \quad  \int_{-\infty}^{\infty} R^{(N)}_{0,1}(x) \, dx = \E[M] \sim \frac{N}{2}.
\]
One then would ask, when suitably scaled and normalized as in
$$
  s^{(N)}_1(x) =  \frac{1}{\sqrt{2}} R^{(N)}_{1,0}(\sqrt{N} x) \quad \mbox{
    and } \quad s_2^{(N)}(x) = \frac{2}{\sqrt{N}} R^{(N)}_{0,1}(\sqrt{N} x),
$$
whether $s_{1}^{(N)}(x) dx $  and $s_{2}^{(N)}(x) dx $ converge to proper probability measures.  This is answered
in the affirmative in Theorem~\ref{thm:7} below.

The previous result shows that, with probability one, for all
$N$ large the number of charge 1 particles is $\sqrt{2N}(1+o(1))$.
This suggests that, in the thermodynamic limit, the statistics
of the charge 2 particles should behave as though there are
no charge 1 particles present, or like a copy of the 
Gaussian Symplectic Ensemble (again, arrived at from
the present ensemble upon setting  $L=0$).  Indeed we find the
scaled density of charge 2 particles approaches the semi-circle law.  

On the other hand, though the charge 1 particles exhibit the same level repulsion 
amongst
themselves as the eigenvalues in the Gaussian
Orthogonal Ensemble (occurring here when $M=0$), the
preponderance of charge 2 particles leads to a different limit distribution.

Ginibre's real ensemble, the ensemble of
eigenvalues of real asymmetric matrices with i.i.d. Gaussian entries,
has superficial resemblance to the ensemble we are considering here.
First, it is suggestive to think of the present ensemble  as arising
from real Ginibre by forcing the non-real eigenvalues, which occur in complex
conjugate pairs, to be identified with one ``charge two'' particle on the line.
A little more concretely,  the (random) number of real
eigenvalues  in real Ginibre has both expectation and variance of $O(\sqrt{N})$, as
does the number of charge 1 particles here. (See 
 \cite{MR1231689} for the mean, and \cite{forrester-2007} for the variance).  
It is perhaps not surprising, therefore, that the limiting scaled density of charge 
1 particles is the same (up
to a constant) as that of the real eigenvalues in Ginibre's real
ensemble \cite{borodin-2008}.

\begin{thm}
\label{thm:7}
As $N \rightarrow \infty$, $s^{(N)}_1$ 
converges weakly in the sense of measures to the uniform law on
$[-\sqrt{2}, \sqrt{2}]$, and $s^{(N)}_2$ converges in the same manner to
the semi-circular law with the same support.   In particular, it is 
proved that
\[
\int e^{itx} s^{(N)}_1(x) dx \rightarrow \frac{1}{\sqrt{2} t} \sin(\sqrt{2} t)
\]
and
\[
\int e^{itx} s^{(N)}_2(x) dx \rightarrow \frac{\sqrt{2}}{t} J_1(\sqrt{2} t),
\]
where the convergence is pointwise.  
\end{thm}

We give an elementary proof of the above, making use of the explicit
skew-orthogonal polynomial system derived below.  Given that the number
of charge 1 particles is $o(N)$, one could undoubtedly make a large deviation
proof along the lines of \cite{MR1465640} or \cite{MR1660943} of 
a stronger version of the second statement: that the random counting
measure of charge 2 particles converges almost surely to the semi-circle law.
However, it is not clear how to use such energy optimization ideas to access
the charge 1 profile.

\begin{figure}[h]
\centering
\includegraphics[scale=.6]{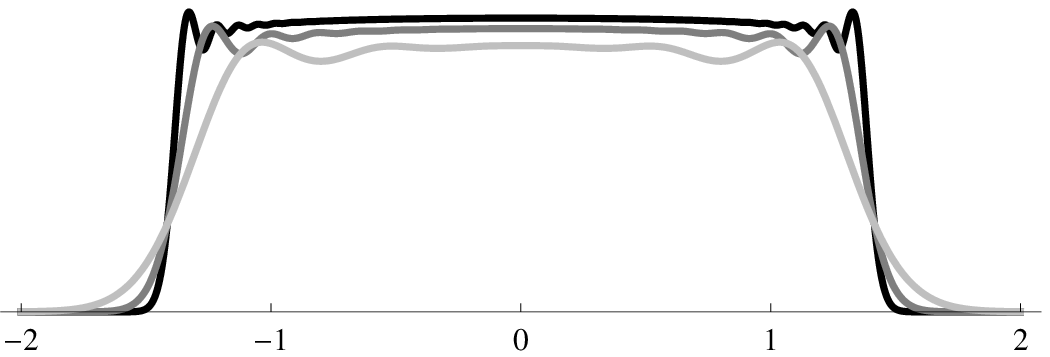}
\hspace{1cm}
\includegraphics[scale=.6]{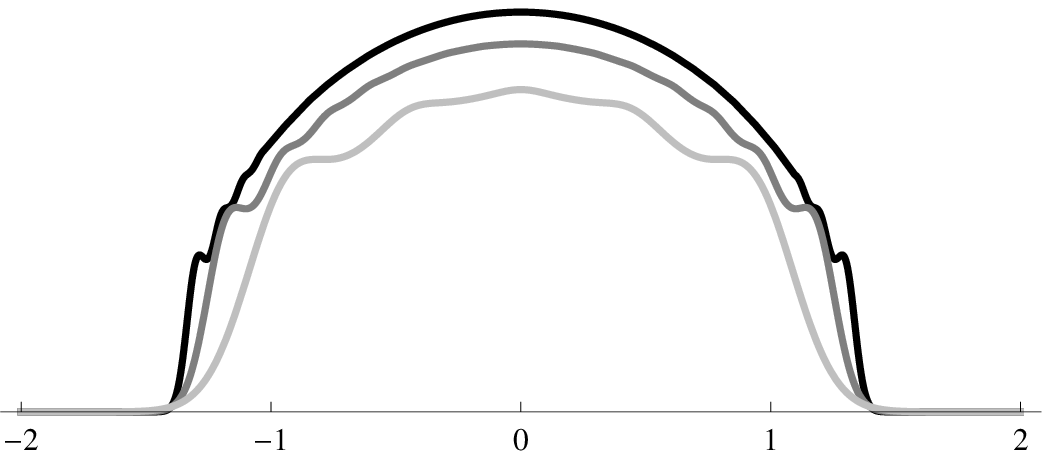}
\begin{caption}{$s_1^{(N)}$ (left) and $s_2^{(N)}$ (right) for, from
    lightest to darkest, $N=10, 30$ and 90.}  
\label{fig:1}
\end{caption}
\end{figure}

\section{A Pfaffian point process for the particles}

All of the results in this paper follow, in one way or another, from
the fact that our interacting particles form a Pfaffian point process
very much like that of Ginibre's real ensemble and related to the
Gaussian Orthogonal and Symplectic Ensembles.  

The results in this section are valid for quite general weight
functions $w$ and fugacities.  Thus, for the time being, we will
return to the general situation.

\subsection{The joint density of particles}

The joint density of particles for a particular choice of $(L,M)$ is given by
\[
\frac{X^{L}}{Z(X)} \Omega_{L,M}(\bs \alpha, \bs \beta), \qquad \mbox{where}
\qquad \Omega_{L,M}(\bs \alpha, \bs \beta) = e^{-E(\bs \alpha, \bs \beta)}.
\]
More specifically,
\[
\Omega_{L,M}(\bs \alpha, \bs \beta) = \prod_{\ell=1}^L
w(\alpha_{\ell}) \prod_{m=1}^M w(\beta_m)^2 \prod_{j < k} | \alpha_j -
\alpha_k | \prod_{m < n} |\beta_m - \beta_n|^4 \prod_{\ell=1}^L
\prod_{m=1}^M |\alpha_{\ell} - \beta_m|^2;
\]
where, for now, the only assumptions we will make on $w$ are that it
is positive and Lebesgue measurable with $0 < Z(X) < \infty$.  

\subsection{Correlation Functions}
\label{sec:corr-funct}

Given $0 \leq \ell \leq 
L$ and $0 \leq m \leq M$, we define the $\ell,m$-correlation function
$R^{(N)}_{\ell,m}: \R^{\ell} \times \R^m \rightarrow [0,\infty)$ by 
\[
R^{(N)}_{\ell, m}(\mathbf x; \mathbf y) = \sum_{(L,M) \atop L \geq \ell, M \geq m}
\frac{1}{(L-\ell)! (M-m)!} \int\limits_{\R^{L-\ell}} \int\limits_{\R^{M-m}}
\Omega_{L,M}(\mathbf x \vee \bs \alpha, \mathbf y \vee \bs \beta) \,
d\mu^{L-\ell}(\bs \alpha) \, d\mu^{M-m}(\bs \beta),
\]
where, for instance, $\mathbf x \vee \bs \alpha$ is the vector in
$\R^L$ formed by concatenating $\mathbf x \in \R^{\ell}$ and $\bs
\alpha \in \R^{L - \ell}$.  We will often write $R_{\ell,m}$ for
$R^{(N)}_{\ell,m}$ in situations where $N$ is seen as being fixed.  

\sloppy{ The correlation functions encode statistical information about the
configurations of the charged particles.} To be more precise, given
$\bs \alpha \in \R^{L}$ and $\bs \beta \in \R^{M}$ with $L + 2 M
= N$, we set
\[
\xi = \xi(\bs \alpha, \bs \beta) = (\xi_1, \xi_2) = \big(\xi_1(\bs
\alpha), \xi_2(\bs \beta) \big) = \big( \{ \alpha_1, \ldots,
\alpha_{\ell} \} , \{ \beta_1, \ldots \beta_m \} \big).
\]
Given an $L$-tuple of mutually disjoint subsets of
$\R$, $\mathbf A = (A_1, A_2, \ldots, A_L)$, and an $M$-tuple of
mutually disjoint subsets of $\R$, $\mathbf B = (B_1, B_2, \ldots, B_M)$, the 
probability that the system is in a state where there is 
exactly one charge 1 particle in each of the $A_{\ell}$ and
exactly one charge 2 particle in each of the $B_m$ is given by 
\begin{align*}
& \mathrm{Prob}\{|A_1 \cap \xi_1| = 1, \ldots, | A_L \cap \xi_1| = 1,
|B_1 \cap \xi_2| = 1, \ldots,  |B_M \cap \xi_2| = 1\} \\
& \hspace{6cm} = \E\left[ \bigg\{ \prod_{\ell = 1}^L
  | A_{\ell} \cap \xi_1 | \bigg\} \bigg\{ \prod_{m=1}^M
  | B_m \cap \xi_2 | \bigg\} \right].
\end{align*}
This probability can also be represented by 
\begin{align*}
\frac{1}{L! M!} \sum_{\sigma \in S_L} \sum_{\tau \in S_M} \int_{B_{\tau(1)}}
\cdots \int_{B_{\tau(M)}} 
\int_{A_{\sigma(1)}} \cdots \int_{A_\sigma(L)}  \Omega_{L,M}(\bs \alpha, \bs
\beta) \,  d\mu^{L}(\bs \alpha) \, d\mu^M(\bs \beta).
\end{align*}
Since the integrand is symmetric in
the coordinates of $\bs \alpha$ and $\bs \beta$, we find
\[
\E\left[ \bigg\{ \prod_{\ell = 1}^L
  | A_{\ell} \cap \xi_1 | \bigg\} \bigg\{ \prod_{m=1}^M
  | B_m \cap \xi_2 | \bigg\} \right] = \int_{\mathbf
  B} \int_{\mathbf A} \Omega_{L,M}(\bs \alpha, \bs \beta) \,
d\mu^{L}(\bs \alpha) \, d\mu^M(\bs \beta).
\]

The correlation functions can be used to generalize this formula. If
$\mathbf A = (A_1, A_2, \ldots, A_{\ell})$ is a tuple of disjoint
subsets of $\R$ and $\mathbf B = (B_1, B_2,
\ldots, B_m)$ another such tuple, then
\[\E\left[ \bigg\{ \prod_{j = 1}^{\ell} |A_j \cap \xi_1| \bigg\} \bigg\{
  \prod_{k=1}^m | B_k \cap \xi_2| \bigg\} \right] = \int_{\mathbf
  B} \int_{\mathbf A} R_{\ell,m}(\mathbf x; \mathbf y) \,
d\mu^{\ell}(\bs \alpha) \, d\mu^m(\bs \beta).
\]

\subsection{Pfaffian point processes}

Consider, for the moment, a simplified system of indistinguishable
random points $\zeta = \{\gamma_1, \gamma_2, \ldots, \gamma_N\}
\subseteq\R$ with correlation functions $R_n(\mathbf z)$ satisfying 
\[
\E \left[ \prod_{j=1}^n | A_j \cap \zeta| \right] =
\int_{A_1} \cdots \int_{A_n} R_n(\mathbf z) \, d\mu^{n}(\mathbf z)
\]
for any $n$-tuple $(A_1, A_2, \ldots, A_n)$ of mutually disjoint
sets.  

If there exists a matrix valued function $K_N : \R^2 \rightarrow \R^{2
  \times 2}$ such that
\[
R_n(\mathbf z) = \Pf \left[ K_N(z_j, z_k) \right]_{j,k=1}^{n},
\]
then we say that our ensemble of random points forms a {\em Pfaffian
  point process} with {\em matrix kernel} $K_N$.  Much of the
information about probabilities of locations of particles ({\em
  e.g.}~gap probabilities) can be derived from properties of the
matrix kernel.  Moreover, in many instances, we are interested in
statistical properties of the particles as their number (or
some related parameter) tends toward $\infty$.  In these instances, it
is sometimes possible to analyze $K_N(x,y)$ in this limit (under,
perhaps, some scaling of $x$ and $y$ dependent on $N$) so that the
relevant limiting probabilities are attainable from this limiting
kernel.  

For the ensemble of charge 1 and charge 2 particles with total charge
$N$, we will demonstrate that the correlation functions have a
Pfaffian formulation of the form, 
\[
R_{\ell, m}(\mathbf x; \mathbf y) = 2^{\ell} \Pf \begin{bmatrix}
K_N^{1,1}(x_j, x_{j'}) & K_N^{1,2}(x_j, z_{k'}) \\
K_N^{2,1}(z_k, x_{k'}) & K_N^{2,2}(x_k, x_{k'})
\end{bmatrix}; \qquad {j,j' = 1,2,\ldots, \ell \atop k,k'=1,2,\ldots m}
\]
where $K_N^{1,1}, K_N^{1,2}, K_N^{2,1}$ and $K_N^{2,2}$ are $2 \times
2$ matrix kernels.  

\subsection{A Pfaffian form for the total partition function}

In order to establish the existence of the matrix kernels we first
need a Pfaffian formulation of the total partition function.  

Given a measure $\nu$ on $\R$ we define the operators
$\epsilon_1^{\nu}$ and $\epsilon_2^{\nu}$ on $L^2(\nu)$ by
\[
\epsilon_1^{\nu} f(x) = \frac12 \int_{\R} f(y) \sgn(y-x) d\nu(y)
\qquad \mbox{and} \qquad \epsilon_2^{\nu} f(y) = f'(y).
\]
(Obviously $\epsilon_2^{\nu}$ does not depend on $\nu$, but it is
convenient to maintain symmetric notation).  Using these inner
products we define 
\[
\la f | g \ra_{b^2}^{\nu} = \int_{\R} \left[ f(x) \epsilon_b g(x) - g(x)
  \epsilon_b f(x) \right] \, d\nu(x), \qquad b = 1,2.
\]
We specialize these operators and inner products for Lebesgue measure $\mu$ by
setting $\epsilon_b = \epsilon_b^{\mu}$ and $\la f | g \ra_{b^2}^{\mu}$.
We also write $\wt f(x) = w(x) f(x)$.   It is easily seen that 
\[
\la \wt f | \wt g \ra_1 = \int_{\R} \left[ \wt
  f(x) \epsilon_{1} \wt g(x) -   \wt g(x) \epsilon_{1} \wt f(x)
\right] \, d\mu(x) = \la f | g \ra_1^{w \mu}.
\]
Similarly, 
\begin{align*}
\la \wt f | \wt g \ra_4 &= \int_{\R} \left[ \wt f(x) \frac{d}{dx} \wt
  g(x) - \wt g(x) \frac{d}{dx} \wt f(x) \right] \, d\mu(x) \\
&= \int_{\R} w(x)^2 \left[ f(x) g'(x) - g(x) f'(x) \right] \, dx = \la
f | g \ra_4^{w^2 \mu}.
\end{align*}

We call a family of polynomials, $\mathbf p = \big(p_0(x), p_1(x),
\ldots, p_{N-1}(x)\big)$, a {\em complete} family of 
polynomials if $\deg p_n = n$.  A complete family of monic polynomials is
defined accordingly.  
\begin{thm}
\label{thm:1}
Suppose $N$ is even and $\mathbf p$ is any complete family of monic
polynomials.  Then,
\[
Z(X) = \Pf \left( X^2 \mathbf A^{\mathbf p} +  \mathbf B^{\mathbf p} \right),
\]
where 
\[
\mathbf A^{\mathbf p} = \left[ \la \wt p_m | \wt p_n \ra_1
\right]_{m,n=0}^{N-1} \qquad \mbox{and} \qquad 
\mathbf B^{\mathbf p} = \left[ \la \wt p_m | \wt p_n \ra_4
\right]_{m,n=0}^{N-1}.
\]
\end{thm}
\begin{cor} 
\label{cor:2}
With the same assumptions as Theorem~\ref{thm:1}, $Z =
  \Pf(\mathbf A^{\mathbf p} + \mathbf B^{\mathbf p})$. 
\end{cor}

\subsection{A Pfaffian formulation of the correlation functions}
\label{sec:pfaff-form-corr}

In order to describe the entries in the kernels $K_N^{1,1}, K_N^{1,2},
K_N^{2,1}$ and $K_N^{2,2}$, we suppose $\mathbf
p$ is any complete family of polynomials and define   
\[
\mathbf C^{\mathbf p} = \mathbf A^{\mathbf p} + \mathbf B^{\mathbf p},
\]
where $\mathbf A^{\mathbf p}$ and $\mathbf B^{\mathbf p}$ are as in
Corollary~\ref{cor:2}.  (Here we are setting $X=1$, though similar
maneuvers are valid for general $X > 0$).  Since we are assuming that $Z = \Pf \mathbf
C^{\mathbf p}$ is  non-zero, $\mathbf C^{\mathbf p}$ is invertible and
we set 
\[
(\mathbf C^{\mathbf p})^{-\transpose} = \left[ \zeta_{j,k} \right]_{j,k=0}^{N-1}.
\]
The $\zeta_{j,k}$ clearly depend on our choice of polynomials.  We
then define
\begin{equation}
\label{eq:20}
\varkappa_N(x,y) = \sum_{j,k=0}^{N-1} \wt p_j(x) \zeta_{j,k} \wt p_k(y).
\end{equation}
The operators $\epsilon_1$ and $\epsilon_2$ operate on $\varkappa_N(x,y)$
in the usual manner.  For instance,
\[
\epsilon_2 \varkappa_N(x,y) = \sum_{j,k=0}^{N-1}  \epsilon_2 \wt p_j(x)
\zeta_{j,k} \epsilon_2 \wt p_k(y) 
\]
and 
\[
\varkappa_N \epsilon_1(x,y) = \sum_{j,k=0}^{N-1}  \wt p_j(x)
\zeta_{j,k} \epsilon_1 \wt p_k(y).
\]
(That is, $\epsilon$ written on the left acts on the
$\varkappa_N(x,y)$ viewed as a function of $x$, etc.).
\begin{thm}
\label{thm:2}
Suppose $N$ is even, $\mathbf p$ is any complete family of polynomials
and $\varkappa_N(x,y)$ is given as in (\ref{eq:20}).  Then,
\[
R_{\ell, m}(\mathbf x; \mathbf y) = 2^{\ell} \Pf \begin{bmatrix}
K_N^{1,1}(x_{j}, x_{j'}) & K_N^{1,2}(x_{j}, z_{k'}) \\
K_N^{2,1}(z_k, x_{k'}) & K_N^{2,2}(x_k, x_{k'})
\end{bmatrix}; \qquad {j,j' = 1,2,\ldots, \ell \atop k,k'=1,2,\ldots
  m},
\]
where
\[
K_N^{1,1}(x,y) = \begin{bmatrix}
\varkappa_N(x, y) &   \varkappa_N \epsilon_1 (x,y) \\
 \epsilon_1 \varkappa_N(x,y) & \epsilon_1 \varkappa_N \epsilon_1(x,y)
 + \frac{1}{4} \sgn(y-x) 
\end{bmatrix},
\]
\[
K_N^{2,2}(x,y) = \begin{bmatrix}
\varkappa_N(x, y) &  \varkappa_N \epsilon_2 (x,y) \\
\epsilon_2\varkappa_N (x,y)  & \epsilon_2 \varkappa_N \epsilon_2(x,y) 
\end{bmatrix},
\]
\[
K_N^{1,2}(x,y) = \begin{bmatrix}
\varkappa_N(x, y) &  \varkappa_N \epsilon_1 (x,y) \\
 \epsilon_2 \varkappa_N (x,y) & \epsilon_2 \varkappa_N \epsilon_1(x,y) 
\end{bmatrix} \quad and \quad K_N^{2,1}(x,y) = \begin{bmatrix}
\varkappa_N(x, y) & \varkappa_N \epsilon_2(x,y) \\
 \epsilon_1 \varkappa_N(x,y) & \epsilon_1 \varkappa_N \epsilon_2(x,y) 
\end{bmatrix}.
\]

\end{thm}

\begin{rem}
The factor $2^{\ell}$ can be moved inside the Pfaffian so that 
the entries in the various kernels where an $\epsilon_1$ appears are
multiplied by 2.  This maneuver is superficial, but has the effect of
making these particular entries appear more like the entries in other
$\beta=1$ ensembles ({\it e.g.} GOE). For instance, $\frac12
\sgn(y-x)$ appears more natural to experts used to these other
ensembles.  
\end{rem}

We can simplify the presentation of the matrix kernels with a bit of
notation.  First, let us write
\[
K_N(x,y) = \begin{bmatrix} \varkappa_N(x,y) & \varkappa_N(x,y) \\
\varkappa_N(x,y) & \varkappa_N(x,y)\end{bmatrix},
 \qquad \mbox{and} \qquad 
E_b = \begin{bmatrix}
1 & 0 \\
0 & \epsilon_b
\end{bmatrix}; \quad b = 1,2.
\]
Then, 
\[
K_N^{1,1}(x,y) = E_1 K_N(x,y) E_1 + \begin{bmatrix} 0 & 0 \\ 0 &
  \frac{1}{4} \sgn(y - x) \end{bmatrix},
\]
\[
K_N^{2,2}(x,y) = E_2 K_N(x,y) E_2, \quad K_N^{1,2}(x,y) = E_1
K_N(x,y) E_2,  \quad  K_N^{2,1}(x,y) = E_2 K_N(x,y)
E_1. 
\]

We notice in particular that the functions $R^{(N)}_{1,0}$ and $R^{(N)}_{0,1}$
given in Section~\ref{sec:spat-dens-part} are given by
\begin{equation}
R^{(N)}_{1,0}(x) = 2 \sum_{j,k=0}^{N-1} \wt p_{j}(x) \zeta_{j,k} \epsilon_1
\wt p_{k}(x) \qquad \mbox{and} \qquad R^{(N)}_{0,1}(x) = \sum_{j,k=0}^{N-1}
\wt p_{j}(x) \zeta_{j,k} \epsilon_2 \wt p_{k}(x).
\end{equation}

\subsection{Skew-orthogonal polynomials}

The entries in the kernel themselves can be simplified (or at least
presented in a simplified form) by a judicious choice of $\mathbf p$.
If we define
\[
\la f | g \ra = \la f | g \ra_1 + \la f | g \ra_4,
\]
then 
\[
\mathbf C^{\mathbf p} = \left[ \la \wt p_m | \wt p_n \ra \right]_{m,n=0}^{N-1}.
\]
Since $\varkappa_N$ (and by extension all other entries of the various
kernels) depend on the inverse transpose of $\mathbf C^{\mathbf p}$,
it is desirable to find a complete family of polynomials for which
$\mathbf C^{\mathbf p}$ can be easily inverted.  

We say $\mathbf p = (p_0, p_1, \ldots )$ is a family of
{\em skew-orthogonal} polynomials for the skew-inner product
$\la \cdot | \cdot \ra$ with weight $w$ if there exists real numbers (called
{\em normalizations}) $r_1, r_2, \ldots $ such that
\[
\la \wt p_{2j} | \wt p_{2k} \ra = \la \wt p_{2j+1} | \wt p_{2k + 1}
\ra = 0 \qquad \mbox{and} \qquad \la \wt p_{2j} | \wt p_{2k+1} \ra = -
\la \wt p_{2k+1} | \wt p_{2j} 
\ra = \delta_{j,k} r_j.  
\]

Using these polynomials, the entries in the matrix kernels presented
in Section~\ref{sec:pfaff-form-corr} have a particularly simple
form.  For instance, 
\begin{align*}
\varkappa_N(x,y) &= \sum_{j=0}^{J-1} \frac{\wt p_{2j}(x) \wt p_{2j+1}(y) -
  \wt p_{2j+1}(x) \wt p_{2j}(y) }{r_j},
\end{align*}
and the entries of the kernels are computed by applying the
appropriate $\epsilon$ operators to this expression.

\subsection{Specification to the Harmonic Oscillator Potential}  

We now return to the case where the weight function is $
w(x) = e^{-x^2/2}$.

\begin{thm} 
\label{thm:5}
Let
\[
\la \cdot | \cdot \ra^{(X)} = X^2 \la \cdot | \cdot \ra_1 + \la
\cdot | \cdot \ra_4.
\]
A complete family of skew-orthogonal polynomials for the weight $w$
with respect to $\la \cdot | \cdot \ra^{(X)}$ is given by
\begin{equation}
\label{skewOP}
P_{2j}^{(X)}(x) = \sum_{k=0}^j (-1)^k \frac{L_k(-X^2)}{L_k(0)}
L_{k}(x^2),
\end{equation}
and
\begin{align}
P_{2j+1}^{(X)}(x) &= 2 x P_{2j}^{(X)}(x) - 2 \frac{d}{dx}
P_{2j}^{(X)}(x) \nonumber \\
\label{skewOPodd2}
   &=  4 X^2 x \sum_{k=0}^{m-1} (-1)^k \frac{L_k^{\frac12}(-
     X^2)}{L_k^{\frac12}(0)} L_k^{\frac12}(x^2) + 2 x
   (-1)^m \frac{L_m^{-\frac12}(- X^2)}{L_m^{-\frac12}(0)}
   L_m^{\frac12}(x^2). 
\end{align}
where $L_k(x) = L_k^{(-1/2)}(x)$ is the generalized $k$th Laguerre
polynomial. The normalization of this family of polynomials is given
by 
\begin{equation}
\label{skewNorm}
\la \wt P^{(X)}_{2m} | \wt P^{(X)}_{2m+1} \ra^{(X)} =  
\frac{4 \pi (m+1)!}{\G{m+\frac12}} L_m(-X^2) L_{m+1}(-X^2).  
\end{equation}
\end{thm}
We can recover a family of monic skew-orthogonal polynomials by
dividing by the leading coefficient.  Specifically,
\begin{cor}
\label{cor:1}
A complete family of monic skew-orthogonal polynomials for the weight $w$
with respect to $\la \cdot | \cdot \ra^{(X)}$ is given by 
\[
p_{2j}^{(X)}(x) = \frac{L_j(0) j!}{L_j(-X^2)} \sum_{k=0}^j (-1)^k
\frac{L_k(-X^2)}{L_k(0)} L_{k}(x^2), 
\]
and 
\[
p_{2j+1}^{(X)}(x) =  x p_{2j}^{(X)}(x) - \frac{d}{dx} p_{2j}^{(X)}(x).
\]
The normalization for this family of monic skew-orthogonal polynomials
is given by
\[
r^{(X)}_j = \la \wt p^{(X)}_{2j} | \wt p^{(X)}_{2j+1} \ra^{(X)} = 4
\frac{(j+1)! \, \G{\! j + \frac12}}{j!} \frac{L_{j+1}(-X^2)}{L_j(-X^2)}.
\]
\end{cor}

Setting $X=1$, we recover a family of skew-orthogonal polynomials for
the harmonic oscillator two charge ensemble with fugacity equal to one, and 
we will write $p_n$ for $p_n^{(1)}$ and $r_j$ for $r_j^{(1)}$.  

\section{Proofs}

\subsection{Proof of Theorem~\ref{thm:6}}

We set $J=N/2$.  To prove 1, we use Theorem~\ref{thm:1} and the skew-orthogonal
polynomials from Corollary~\ref{cor:1} to write 
\[
Z(X) = \Pf \begin{bmatrix}
  0 & r_0^{(X)} \\
  -r_0^{(X)} & 0 \\
& & \ddots \\
& & & 0 & r_{J-1}^{(X)} \\
& & & -r_{J-1}^{(X)} & 0 \\
\end{bmatrix} = \prod_{j=0}^{J-1} r_j^{(X)}.
\]
Hence,
\[
\frac{Z(X)}{Z} = \prod_{j=0}^{J-1} \frac{r_j^{(X)} }{r_j} =
\frac{L_J(-X^2) L_0(-1)}{L_J(-1) L_0(-X^2)} = \frac{L_J(-X^2)}{L_J(-1)},
\]
where again $L_J(x) = L_J^{(-1/2)}(x)$.
Note $L_0(X) = 1$. 

The remaining claims follow from the above by definition and the properties of  Laguerre polynomials.

\subsection{Proof of Theorem~\ref{thm:4}}  
Point 3 of Theorem~\ref{thm:6} specified to the first two moments produces
$$
   \E[L]  = \frac{d}{dX}\left[
      \frac{Z(X)}{Z} \right]_{X=1} ,
  \   \  \Var(L) =
\left[ \frac{d}{dX} \left( X \frac{d}{dX} \frac{Z(X)}{Z} \right) - \left(
  \frac{d}{dX} \frac{Z(X)}{Z} \right)^2 \right]_{X = 1}.
$$
Now, since  $L_J'(x) = - L_{J-1}^{1/2}(x)$ and $L_{J}(x) =   L_J^{1/2}(x) - L_{J-1}^{1/2}(x)$,  we have that
$$
   \mathbb{E} (L) =  2 \frac{L_{J-1}^{1/2}(-1)}{L_J(-1)} =  2  \frac{L_J^{1/2}(-1)}{L_J(-1)}  -2.  
$$ 
Further, using the differential equation $x L_J''(x) + (1/2- x)L_J'(x) + J L_J(x)=0$, we also have that
\begin{align*}
   \frac{d}{dx} \left( x \frac{d}{dx} L_J(- x^2)   \right) & = - 4x L_J'(-x^2) +  4x^3 L_J''(-x^2) \\
     & = (-2x  + 4 x^3) L_J'(-x^2) +  4x J L_J(-x^2). 
\end{align*}
This yields
$$
   \Var(L) = 4 J -   \mathbb{E}(L)  -  \mathbb{E}(L)^2,
$$
and so asymptotics of the variance follow from those for the mean.

Next introduce a version of Perron's formula (see \cite{MR2448665}), 
$$
   L_n^\alpha (-1) = \frac{1}{2 \sqrt{\pi e}} m^{\alpha/2 - 1/4} e^{2 \sqrt{m}} \left(1 + C_1(\alpha) m^{-1/2} +
    C_2(\alpha) m^{-1} +  \mathrm{O}(m^{-3/2})\right), 
$$
where $m = n+1$ and $C_j(\alpha)$ are known explicitly. In particular, $C_1(1/2) = - 1/6$, $C_2(1/2) = -7/144$, 
$C_1(-1/2) = -2/3$, and $C_2(-1/2) = 77 /144$.  Substituting  into the above  we then obtain
$$
\mathbb{E} (L) = 2 \sqrt{J+1}  - {1} - \frac{2}{3 \sqrt{J+1}} + \mathrm{O}(J^{-1})
   = 2 \sqrt{J}  - {1} + \frac{1}{3\sqrt{J}} + \mathrm{O}(J^{-1}),
$$
and $\Var(L) =  {2} \sqrt{J} - \frac{4}{3} + O(J^{-1/2})$ which completes the proof of point 1 (recall $J=N/2$).

Moving to the limit law for $L$, we introduce a little new notation. Set
$$
   p_N(k) =  \frac{C_N}{ \Gamma( \frac{N}{2} - \frac{k}{2} +1 ) \Gamma (\frac{k}{2} + \frac{1}{2} ) \Gamma( \frac{k}{2} +1) } = C_N q_N(k)^{-1}
$$ 
with $C_N = \Gamma(\frac{N}{2} +\frac{1}{2}) [ L_{N/2} (-1)]^{-1} $.  For $k$ even, $p_N(k)$ is the probability of $k$ particles of charge 1, otherwise
this probability is zero, compare point 1 of Theorem~\ref{thm:6}.
In the continuum limit this distinction is unimportant; we will show that, as $ N \rightarrow \infty$ 
\begin{equation}
\label{localCLT}
      (2N)^{1/4} \, p_N \Bigl( (2N)^{1/2} + (2N)^{1/4} c \Bigr) =   \frac{e^{-c^2/2}}{\sqrt{2 \pi} } (1 + O 
      ( {N^{-1/4}}))
\end{equation}
uniformly for $c$ on compact sets. 

First note that by Stirling's approximation (in the form $\Gamma(z) = \sqrt{\frac{2\pi}{z}} (z/e)^z   (1+O(\frac{1}{z}))$) and again
 Perron's formula (now in the simpler form $L_z(-1) = \frac{1}{2 \sqrt{\pi e z}} e^{2 \sqrt{z}} (1 +  O(\frac{1}{\sqrt{z}}) )$), 
\begin{equation}
\label{CN}
   C_N = 2 \pi \sqrt{N e} \, (N/2)^{(N/2)} e^{-N/2 - \sqrt{2N} } (1+O ( N^{-1/2} ) ).
\end{equation}
Next, with both $k$ and $N-k$ large we have
\begin{eqnarray}
\label{qNk}
  q_N(k)  & = &  (2\pi)^{3/2}  \sqrt{Nk}  \, (N/2)^{(N/2)} e^{ - N/2 - k/2 }   \\
   &  &   \,\,  \times \, e^{[ (N/2 -k/2) \log (1 - k/N) + (k/2) \log( k^2/2N) ]} 
             ( 1 + O (  k^{-1}  \vee  (N-k)^{-1} \vee k N^{-1})   ), \nonumber 
\end{eqnarray}
again by Stirling's approximation.  Restricting to $k = O(\sqrt{N})$, 
(\ref{CN}) and (\ref{qNk}) yield
\begin{eqnarray}
\label{pNk}
  p_N(k)  
                & =  &     \sqrt{ \frac{e}{2\pi k} } e^{ - \phi_N(k)}  ( 1 + O (N^{-1/2}) ),
\end{eqnarray}
where
$$
 \phi_N(k) = \sqrt{2N} -  \frac{k}{2}  +  (\frac{N}{2} -\frac{k}{2} ) \log (1 - \frac{k}{N}) + \frac{k}{2} \log( \frac{k^2}{2N} ).
$$
 Now, quite simply
$$
\label{log1} 
   ( \frac{N}{2} - \frac{k}{2} ) \log (1- \frac{k}{N}) =  - \frac{k}{2} + \frac{k^2}{4 N} + O(N^{-1/2}),
$$
if $k = O(\sqrt{N} )$,  and, if $k$ is also such that  $ 1 - \frac{k^2}{2N} = O({N^{-1/4}})$, we further have
$$
\label{log2}
    \frac{k}{2} \log (\frac{k^2}{2N}) =  - \frac{k}{2} ( 1 - \frac{k^2}{2N}) - \frac{k}{4} ( 1 - \frac{k^2}{2N})^2 + O( {N^{-1/4}}). 
$$
More precisely,  from the last two displays we readily find that
$$
  \phi_N( \sqrt{2N} +\ell ) = \frac{1}{2} +  \frac{\ell^2}{2 \sqrt{2 N}} + O(N^{-1/4}),  \mbox{ uniformly for } \ell  = O(N^{1/4}). 
$$
Substituting back into (\ref{pNk}), since $ ( \sqrt{2N} + \ell)^{-1/2} =
(2N)^{-1/4} ( 1 + O(N^{-1/4}))$ again for  $\ell = O(N^{-1/4})$, completes the verification of
(\ref{localCLT}).

Last, for the tail estimate, revisiting (\ref{CN}) and (\ref{qNk}) shows
the conclusion of (\ref{pNk}) may be modified to read
$$
    C^{-1} k^{-1/2}  e^{-\phi_N(k)} \le  p_N(k)  
                \le       C e^{ - \phi_N(k)},
$$
for all $1 \le k \le N$
with a numerical constant $C$.   (Here we understand $(1-\frac{k}{N} )\log(1 -\frac{k}{N})$ to be zero
at $k=N$.)  Differentiating yields 
$$
  \frac{d}{dk} \phi_N(k) = \frac{1}{2} \log  \left( \frac{k^2}{2 (N-k)} \right) ,
$$
and so 
$\phi_N(k)$   is decreasing for $k < c^{-1} \sqrt{2n}$ and increasing for $k > c \sqrt{2N}$
for any $c >1$.
Now, since $(1-\epsilon) \log (1-\epsilon) \ge -\epsilon$ and $\log(1+\epsilon) \ge \epsilon -\epsilon^2/2$ for 
$0< \epsilon \le 1$,
\begin{eqnarray*}
 \phi_N( (1+\epsilon) \sqrt{2N})  & \ge  &  - \epsilon \sqrt{2N} + 2 (1+\epsilon) \sqrt{2N} \log (1+\epsilon)
   \ge    \epsilon \sqrt{2N},
\end{eqnarray*}
also for $0< \epsilon \le 1$.
Hence, for $c > 1$,  $\Prob( L > c \sqrt{2 N})  \le  N p_N( c
\sqrt{2N}) \le C  N e^{- ((c-1)\wedge 1) \sqrt{2 N} }$.   The proof for the left tail is much the same.

\subsection{Proof of Theorem~\ref{thm:7} }

In both cases we use the expression of the one point function in terms of Hermite polynomials,
see (4.17) and (4.20) below. 

We start with
\begin{equation*}
\label{SN1}
   s_N^{(1)}(x) =    \sqrt{2} \sum_{n=0}^{N/2-1}  \frac{ \epsilon_1 \wt{p}_{2n+1}(x)  \wt{p}_{2n}(\sqrt{N}x) 
                                                  -    \wt{p}_{2n+1} ( \sqrt{N} x)  \epsilon_1 \wt{p}_{2n}(\sqrt{N}x)}{r_n},
\end{equation*}
and
\begin{equation*}
\label{SN2}
s_N^{(2)}(x)
 =   \frac{2}{\sqrt{N}} 
 \sum_{n=0}^{N/2-1} \frac{ \wt{p}_{2n+1}'({\sqrt{N}}x) \wt{p}_{2n}({\sqrt{N}}x) - \wt{p}_{2n+1}({\sqrt{N}}x) \wt{p}_{2n}'({\sqrt{N}}x)}{r_n}, 
\end{equation*}
along with the relations $\int_{-\infty}^x \wt{p}_{2n+1} = \epsilon_1 \wt{p}_{2n+1}(x)  =  2 \wt{p}_{2n}(x)$
and $ \wt{p}_{2n}'(x) =  \epsilon_2  \wt{p}_{2n}(x) = -\frac{1}{2} \wt{p}_{2n+1}(x)$.  An integration by parts in
both instances then allows: with $t_N = t/\sqrt{N}$,
\begin{eqnarray}
\label{Fourier1}
\lefteqn{ \int_{-\infty}^{\infty} e^{itx} s_N^{(1)}(x) dx    } \\
& = & \frac{4\sqrt{2}}{\sqrt{N}}  \sum_{n=0}^{N/2-1} r_{n}^{-1} \int_{-\infty}^{\infty} e^{i {t_N} x} ( \wt{p}_{2n}(x) )^2 dx 
        - \frac{2 \sqrt{2}it_N}{\sqrt{N}}  \sum_{n=0}^{N/2-1} r_{n}^{-1} \int_{-\infty}^{\infty} e^{i t_{N} x}   \wt{p}_{2n}(x)  
         \epsilon_1 \wt{p}_{2n}(x)  dx, 
        \nonumber 
\end{eqnarray}
and 
\begin{eqnarray}
\label{Fourier2}
\lefteqn{ \int_{-\infty}^{\infty} e^{itx} s_N^{(2)}(x) dx     } \\
& =& \frac{2}{N}  \sum_{n=0}^{N/2-1} r_{n}^{-1} \int_{-\infty}^{\infty} e^{i {t_N} x} ( \wt{p}_{2n+1}(x) )^2 dx 
        - \frac{2it_N}{N}  \sum_{n=0}^{N/2-1} r_{n}^{-1} \int_{-\infty}^{\infty} e^{i t_{N} x}  \wt{p}_{2n}(x)  \wt{p}_{2n+1}(x)  dx. \nonumber
\end{eqnarray} 
The first, and primary, step is to show that the advertised limits stem from the first sums on the right of the above expressions.

\begin{lemma}  Let $\hat{s}_N^{(1)}(t)$ and $\hat{s}_N^{(2)}(t)$ denote, respectively, the first term on the right hand side
of (\ref{Fourier1}) and (\ref{Fourier2}).  Then,
$$
    \hat{s}_N^{(1)}(t) \rightarrow \frac{ \sin{ \sqrt{2} t}}{\sqrt{2} t}, \ \ \    \hat{s}_N^{(2)}(t) \rightarrow  \frac{\sqrt{2}}{t} J_1(\sqrt{2} t)     
$$ 
as $N \rightarrow \infty$.
\end{lemma}

\begin{proof}
Since $\wt{H}_k(x) = H_k(x) e^{-x^2/2}$ are the eigenfunctions of the Fourier transform $-$ in particular $ (\wt{H}_k)^{\, \widehat{}} (x) $ $= 
\frac{1}{\sqrt{2\pi} }  \int_{-\infty}^{\infty} e^{i x u} \wt{H}_k(u) du$
$=  i^{k} H_{k}(x)$ $-$ we have that
\begin{eqnarray}
\label{FTS}
     \widehat{ \wt{p}_{2n} }(x)   & = &      \sum_{k=0}^n (-1)^k a_k  \wt{H}_{2k}(x),  \\
     \widehat{ \wt{p}_{2n+1}} (x) & =  & 
     i  \sum_{k=0}^n (-1)^k a_k ( \wt{H}_{2k+1}(x) + 4k \wt{H}_{2k-1}(x)) = 2i x \sum_{k=0}^n (-1)^k a_k  \wt{H}_{2k}(x).   \nonumber
\end{eqnarray}
The last equality makes use of the three term recurrence $H_{n+1}(x) = 2x H_{n}(x) - 2n H_{n-1}(x)$.
Plancheral's identity then yields,
\begin{equation}
\label{Lead_Case1} 
\hat{s}_N^{(1)}(t) = \frac{4\sqrt{2}}{\sqrt{N}}  \sum_{n=0}^{N/2-1}  r_n^{-1} \sum_{0 \le k, \ell \le n}   (-1)^{k+\ell}  a_k a_{\ell}
    \int_{-\infty}^{\infty}  \wt{H}_{2k}(x +t_N/2)  \wt{H}_{2 \ell}(x-t_N/2) \, dx,
\end{equation}
and 
\begin{equation}  
\label{Lead_Case2}
  \hat{s}_N^{(2)}(t) = \frac{8}{N}  \sum_{n=0}^{N/2-1}  r_n^{-1} \sum_{0 \le k, \ell \le n}   (-1)^{k+\ell}  a_k a_{\ell}
    \int_{-\infty}^{\infty} (x^2- (t_N/2)^2) \wt{H}_{2k}(x +t_N/2)  \wt{H}_{2 \ell}(x-t_N/2) \, dx. 
\end{equation}
We begin with the asymptotic considerations of (\ref{Lead_Case1}) which is slightly simpler.

From the expansion $H_n(a+b) = \sum_{k=0}^n {n \choose k}  H_k(a) (2b)^{n-k}$ we find that
\begin{eqnarray*}
  \int_{-\infty}^{\infty} \wt{H}_{2k}(x +t/2)  \wt{H}_{2 \ell}(x-t/2) \, dx 
& = &  e^{-t^2/4} \sqrt{\pi} \sum_{m=0}^{2(k \wedge \ell)} { 2k \choose m} {2 \ell \choose m} m! (-2)^m t^{2k+2\ell -2m}.
\end{eqnarray*}
Given this, $\hat{s}_N^{(1)}$ is equivalent, as $N \rightarrow \infty$, to
\begin{eqnarray}
\label{firstsplit}
\lefteqn{
 \hat{s}_{N,d}^{(1)} + \hat{s}_{N,o}^{(1)}  =    
  \frac{4\sqrt{2 \pi}}{\sqrt{N}}  \sum_{n=0}^{N/2-1}  r_n^{-1} \sum_{k=0}^{n}  a_k^2  2^{2k}(2k)! \sum_{m=0}^{2k} \frac{(2k)_m}{(m!)^2}  
     \left(-\frac{t^2}{2N} \right)^m} \\
     &   & + \frac{8\sqrt{2\pi}}{\sqrt{N}}  \sum_{n=1}^{N/2-1}  r_n^{-1} \sum_{0 \le k < \ell \le n}
             a_k a_{\ell} 2^{k+\ell} (2k)! \sum_{m=0}^{2k}  \frac{(2\ell)_{2 \ell-2k +m}}{m! (2\ell-2k+m)!} 
             \left(-\frac{t^2}{2N} \right)^{\ell-k+m}, 
     \nonumber
\end{eqnarray}
in which we have introduced a self-evident notation for the diagonal and off-diagonal components as well
as the (nontraditional) shorthand $(n)_m := \frac{n!}{(n-m)!}$.

Next, recall the definitions
$
   a_n = \frac{n!}{(2n)!} L_n^{(-1/2)}(-1), $ $ r_n = \sqrt{\pi}  2^{2n+2} (2n+2)! a_n a_{n+1}
$
and note the simple appraisals: with $ c =  (4\pi e)^{-1}$,
\begin{eqnarray}
\label{Rasymp}
  r_n^{-1} &   = &  \frac{1}{4 c \pi  \sqrt{ n }}  \,  e^{-4\sqrt{n}}  (1 + O(n^{-1/2})), \\
    a_n a_{n+m}   2^{2n+2m} (2n)!  & =  & c \sqrt{\pi}  n^{-m-1/2}  e^{4\sqrt{n}}   (
  1 + O( m n^{-1/2} )), \nonumber
\end{eqnarray}
where the latter will be used for $m$ nonnegative and  moderate (compared with $n^{1/2}$).  We will
also make repeated use of the fact
\begin{equation}
\label{SimpleSum}
    \sum_{k=1}^n  n^{m-1/2} e^{4 \sqrt{n}} = \frac{1}{2} n^{m} e^{4 \sqrt{n}} (1+ O(n^{-1/2})), 
\end{equation}
valid for any real $m$.

Continuing, we change the order of summation to write 
$$
    \hat{s}_{N,d}^{(1)}  = \frac{4\sqrt{2 \pi}}{\sqrt{N}}  \sum_{m=0}^{N-2}  \frac{(-2t^2/N)^m}{(m!)^2}  
    \sum_{n=  \lceil m/2 \rceil}^{N/2-1} r_n^{-1} \sum_{k = \lceil m/2 \rceil}^n (2k)_m   a_k^2  2^{2k} (2k)!.
$$
Then, for fixed $m$,
\begin{eqnarray*}
   \sum_{n=  \lceil m/2 \rceil}^{N/2-1} r_n^{-1} \sum_{k = \lceil m/2 \rceil}^n (2k)_m   a_k^2 2^{2k} (2k)!  
   & = & \frac{N^{m+1/2}}{4 \sqrt{2 \pi} (2m+1)} (1+o(1)),
\end{eqnarray*}
by (\ref{Rasymp}) and (\ref{SimpleSum}), and a dominated convergence argument yields
\begin{equation}
\label{diag}
 \lim_{N \rightarrow \infty}  \hat{s}_{N,d}^{(1)} =  \sum_{m=0}^{\infty} \frac{(-  t^2/2)^m}{(m!)^2(2m+1)} = \int_0^1 J_0(\sqrt{2} t x) dx,
\end{equation}
for the diagonal contribution.
Next, for the off-diagonal terms (second line of (\ref{firstsplit})),
we again change the order of summation and have $ \hat{s}_{N,o}^{(1)}$
equal to 
\begin{eqnarray*}
&&\frac{8 \sqrt{2 \pi}}{ \sqrt{N}} \sum_{q=1}^{N/2-1} \sum_{m=0}^{N-2q-2}  \frac{ (-2t^2/N)^{q+m}}{m! (2q+m)!}  
                                      \sum_{n = q+\lceil m/2 \rceil}^{N/2-1} r_n^{-1} \sum_{k= \lceil m/2 \rceil}^{n-q} (2k+2q)_{2q+m} a_k a_{k+q} (2k)!
                                                      2^{2k+q}.
\end{eqnarray*}
With now $q$ and $m$ fixed,
\begin{eqnarray*}
\lefteqn{
    \sum_{n = q+\lceil m/2 \rceil}^{N/2-1}  r_n^{-1}\sum_{k= \lceil m/2 \rceil}^{n-q} (2k+2q)_{2q+m} a_k a_{k+q} (2k)!
                                                      2^{2k+q} }\\ & = &  \frac{ 2^{q+m}}{4 \sqrt{\pi}} \sum_{n=1}^{N/2} \frac{1}{\sqrt{n}} e^{-4\sqrt{n}}
                                                          \sum_{k=1}^{n} k^{q+m -1/2} e^{4 \sqrt{k}} (1+o(1)) 
                                                            =   \frac{N^{q+m+1/2}}{4 \sqrt{2 \pi} (2q+2m+1)} (1+o(1)),
\end{eqnarray*}
again by  (\ref{Rasymp}) and (\ref{SimpleSum}).
Hence, for bounded $t$,
\begin{eqnarray*}
    \hat{s}_{N,o}^{(1)}  &  =  &  2 \sum_{q=1}^{N/2-1} \sum_{m=0}^{N-2q-2}  \frac{ (-t^2/2)^{q+m}}{m! (2q+m)! (2q+2m+1)}  (1+o(1)) \\
                                       & =   &  2 \sum_{\ell=1}^N  \frac{(-t^2/2)^{\ell}}{(2 \ell +1)} \sum_{q=1}^{\ell}  \frac{1}{(\ell -q)! (\ell + q)!}  (1+o(1)) \\ 
                                        & = &   \sum_{\ell=1}^N  \frac{(-t^2/2)^{\ell}}{(2 \ell +1)} \left( \frac{2^{2\ell}}{(2\ell)!} - \frac{1}{(\ell !)^2} \right) (1+o(1)),
\end{eqnarray*}
after changing variables and the order of summation in line two.  
That is, $  \hat{s}_{N,o}^{(1)}$ tends to $ \frac{\sin{ \sqrt{2} t}}{\sqrt{2} t} - \int_0^1 J_0 (\sqrt{2} t x) dx$,
which, combined with (\ref{diag}), proves the first statement of the lemma.

Turning to (\ref{Lead_Case2}), the preceding shows that, asymptotically, the $(x^2 - (t_N/2)^2)$ within the integrand may be replaced by
$ \frac{1}{4} H_2(x) = x^2-1/2$ for which there is the related evaluation: assuming $k \le \ell$,
\begin{eqnarray}
\label{HermiteEval}
\lefteqn{ \int_{-\infty}^{\infty} H_2(x) {H}_{2k}(x +t/2) {H}_{2 \ell}(x-t/2) e^{-x^2} \, dx  } \\ 
    & = & 4 \sqrt{\pi}  \mathop{\sum_{{0  \le n \le  2k}}}_{{n-m=0,\pm 2}}
      { 2k \choose n}  {2 \ell \choose m}   (-1)^m 
        \frac{n! m!  2^{\frac{n+m}{2}}   t^{2k+2\ell-n-m}}{ ( \frac{n-m}{2}+1)!  ( \frac{m-n}{2}+1)!  (\frac{n+m}{2}-1)!}.
           \nonumber
\end{eqnarray}
The resulting diagonal term (when $k=\ell$ in (\ref{Lead_Case2})) then reads
\begin{eqnarray}
\label{diag2}
   \hat{s}_{N,d}^{(2)} & = & 
   \frac{8 \sqrt{\pi}}{N} \sum_{n=1}^{N/2-1}  r_n^{-1} \sum_{k=1}^{n}  a_k^2  2^{2k}(2k)! \sum_{m=0}^{2k} \frac{(2k)_m (2k-m)} {(m!)^2}  
     \left(-\frac{t^2}{2N} \right)^m \\
     &   & - \frac{8 \sqrt{\pi}}{{N}}  \sum_{n=1}^{N/2-1}  r_n^{-1} \sum_{ k =1}^{n}
             a_k^2  2^{2k} (2k)!  \sum_{m=0}^{2k-2}  \frac{(2k)_{m+2}}{m! (m+2)!} 
             \left(-\frac{t^2}{2N} \right)^{m+1}. \nonumber
\end{eqnarray}
This object does not converge on its own; cancellations from the off-diagonals are required.

With similar notation to the above we  decompose $\hat{s}_{N,o}^{(2)}$ as in $ \sum_{p \ge 1} \hat{s}_{N, (o,+p)}^{(2)}$ in which
$ \hat{s}_{N, (o,+p)}^{(2)}$ is arrived at by choosing $\ell = k+p$ in (\ref{Lead_Case2}).  Writing out the $p=1$ case in full we have that
\begin{eqnarray}
\label{off2}
    \hat{s}_{N, (o,+1)}^{(2)} & = & 
   - \frac{16 \sqrt{\pi}}{N} \sum_{n=1}^{N/2-1}  r_n^{-1} \sum_{k=1}^{n-1}  a_k a_{k+1}  2^{2k}(2k)! \sum_{m=0}^{2k} \frac{(2k+2)_{m+2}}{(m!)^2}   
   \left(-\frac{t^2}{2N} \right)^m  \\
  &  & 
 + \frac{32 \sqrt{\pi}}{N} \sum_{n=1}^{N/2-1}  r_n^{-1} \sum_{k=1}^{n-1}  a_k a_{k+1}  2^{2k}(2k)! \sum_{m=0}^{2k-1} \frac{(2k+2)_{m+2}(2k-m)}{m! (m+2)!}   
   \left(-\frac{t^2}{2N} \right)^{m+1}   \nonumber \\
   &  & -  \frac{16 \sqrt{\pi}}{N} \sum_{n=1}^{N/2-1}  r_n^{-1} \sum_{k=1}^{n-1}  a_k a_{k+1}  2^{2k}(2k)! \sum_{m=0}^{2k-2} \frac{(2k+2)_{m+4}}{m! (m+4)!}   
   \left(-\frac{t^2}{2N} \right)^{m+2}.  
   \nonumber
\end{eqnarray}

Consider now the first sum on the right of (\ref{diag2}) for $k=n$ only:
\begin{eqnarray*}
 \lefteqn{  \frac{8 \sqrt{\pi}}{N} \sum_{n=1}^{N/2-1}  r_n^{-1}  a_n^2  2^{2n}(2n)! \sum_{m=0}^{2n} (2k)_m (2n-m)  \frac{(-t^2/2N)^m}{(m!)^2}} \\  
   & =  &    \frac{8 \sqrt{\pi}}{N}  \sum_{m=0}^{N-2}  \frac{(-t^2/2N)^m}{(m!)^2} \sum_{n = \lceil m/2 \rceil}^{N/2-1}   r_n^{-1}  a_n^2  2^{2n}(2n)! 
          (2n)^{m+1} (1+o(1))  \\
    & = &   2       \sum_{m=0}^{N}  \frac{(-t^2/2)^m}{(m!)^2 (m+1)}   (1+o(1)),
\end{eqnarray*}   
by the same type of estimates used in the analysis of $\hat{s}_N^{(1)}$.
Next, using the additional fact that 
$
        1 - 4 k a_{k+1}{a_k^{-1}} =  - k^{-1/2}(1 + O (k^{-1}))
$
the remainder (or $k \le n-1$ part) of the first sum in (\ref{diag2}) plus the first sum in (\ref{off2}) is asymptotic to
\begin{eqnarray*} 
&&  - \frac{ 8 \sqrt{\pi}}{N} \sum_{m=0}^{N-2}
\frac{(-t^2/2)^m}{(m!)^2}  \sum_{n=  \lceil m/2 \rceil}^{N/2} r_n^{-1}
\sum_{k=\lceil m/2 \rceil}^{n-1} a_k^2 2^{2k} (2k)!   \times
\frac{1}{\sqrt{k}}   (2k)^{m+1} \\  && \hspace{2cm} =   - \sum_{m=0}^{N}
\frac{(-t^2/2)^m}{(m!)^2(m+1)} (1+o(1)). 
\end{eqnarray*} 
The last two displays combine to produce  the advertised limit $ \frac{\sqrt{2}}{t} J_1(\sqrt{2}t )$.

The above ideas propagate.  In particular, the remaining terms of $ \hat{s}_{N,d}^{(2)} +  \hat{s}_{N, (o,+1)}^{(2)} $ balance to produce
a $o(1)$ contribution, and this appraisal extends to the full sum over $p> 1$ of $\hat{s}_{N, (o,+p)}^{(2)}$.  We do not reproduce
the details.
\end{proof}

Revisiting second terms in (\ref{Fourier1}) and (\ref{Fourier2}) shows that
the proof of Theorem \ref{thm:7} can be completed by the following (rough) overestimates.

\begin{lemma} As $N \rightarrow \infty$,
$$
    \sum_{n=0}^{N/2-1} r_n^{-1} \int_{-\infty}^{\infty} | \wt{p}_{2n}(x) \epsilon_1 \wt{p}_{2n}(x) | dx   = O(N^{1/2}),      \sum_{n=0}^{N/2-1} r_n^{-1} \int_{-\infty}^{\infty} | \wt{p}_{2n}(x) \wt{p}_{2n+1}(x) | dx   = O(N).
$$
\end{lemma}

\begin{proof}  Along with the well known evaluation $|| \wt{H}_k||_{L^2} = \pi^{1/4} 2^{n/2} \sqrt{n!}$ used several time already, 
 it holds that $ || \wt{H}_k||_{L^1} 
= c \, 2^{n/2} \sqrt{n!} \, n^{-1/4} (1+O(n^{-1})$ with a (known) numerical constant $c$.  Next note that $ || f \epsilon_1 g ||_{L^1} \le || f ||_{L^1}  || g||_{L^1}$ and so
\begin{eqnarray*}
  ||  \wt{p}_{2n}  \epsilon_1 \wt{p}_{2n} ||_{L^1}   \le  ||  \wt{p}_{2n}  ||_{L^1}^2  
                                                                                     &  = &  \bigl( \sum_{k=0}^n  a_k || \wt{H}_{2k} ||_{L^1} \bigr)^2 \\
                                                                                     & \le & C (  \sum_{k=0}^n \frac{1}{\sqrt{k}} e^{2 \sqrt{k}} )^2 = O ( e^{4 \sqrt{n}}). 
\end{eqnarray*}
Here we have used, again, that $a_k  2^k \sqrt{2k!} \sim k^{-1/4} e^{2 \sqrt{k}}$.  Recalling that $r_n^{-1} \sim n^{-1/2} e^{4 \sqrt{n}}$
finishes the first part.

For the second estimate we can get by with an application of Schwarz's inequality.  Simply compute
$$
  || \wt{p}_{2n} ||_{L^2}^2  =  \sum_{k=0}^n a_k^2 || \wt{H}_{2k}||_{L^2}^2 \le C \sum_{k=0}^n \frac{1}{\sqrt{k}} e^{4\sqrt{k}}  = O ( e^{4 \sqrt{n}}),
$$
and
\begin{eqnarray*}
  || \wt{p}_{2n+1} ||_{L^2}^2 & =  &  \sum_{k=0}^n a_k^2 ( || \wt{H}_{2k+1}||_{L^2}^2 + 16k^2 || \wt{H}_{2k-1}||_{L^2}^2  )
      + 2 \sum_{k=0}^{n-1} a_k a_{k-1}  || \wt{H}_{2k+1}||_{L^2}^2 \\
& \le & C  \sum_{k=0}^n {\sqrt{k}} e^{4\sqrt{k}} = O (n e^{4 \sqrt{n}}), 
\end{eqnarray*}
to find that $r_n^{-1}   || \wt{p}_{2n} ||_{L^2}   || \wt{p}_{2n+1} ||_{L^2} = O(1)$, which suffices.
\end{proof}

\subsection{Proof of Theorem~\ref{thm:1}}

\label{sec:proof-theor-refthm:1}

We will prove something slightly more general which will be useful in
the sequel.

Given measures $\nu_1$ and $\nu_2$ on $\R$, define
\begin{equation}
\label{eq:5}
Z^{\nu_1, \nu_2}_{L,M} = \frac{1}{L! M!} \int_{\R^L} \int_{\R^M}
\bigg\{ \prod_{j < k} | \alpha_j - \alpha_k | \prod_{m < n} |\beta_m -
\beta_n|^4 \prod_{\ell=1}^L \prod_{m=1}^M |\alpha_{\ell} - \beta_m|^2 \bigg\}
d\nu_1^L(\alpha) \, d\nu_2^M(\beta),
\end{equation}
and
\begin{equation}
\label{eq:6}
Z^{\nu_1, \nu_2}(X) = \sum_{(L,M)} X^{L} Z_{L,M}^{\nu_1, \nu_2}.
\end{equation}
\begin{thm}
\label{thm:3}
Suppose $N$ is even and $\mathbf p$ is any complete family of monic
polynomials.  
\begin{equation}
\label{eq:19}
Z^{\nu_1, \nu_2}(X) = \Pf \left( X^2 \mathbf A^{\mathbf p} + \mathbf
  B^{\mathbf p} \right), 
\end{equation}
where 
\[
\mathbf A^{\mathbf p} = \left[ \la p_m | p_n \ra^{\nu_1}_1
\right]_{m,n=0}^{N-1} \qquad \mbox{and} \qquad 
\mathbf B^{\mathbf p} = \left[ \la p_m | p_n \ra^{\nu_2}_4
\right]_{m,n=0}^{N-1}.
\] 
\end{thm}

\subsubsection{The Confluent Vandermonde Determinant} 

A special case of the confluent Vandermonde determinant identity has
that
\begin{align}
& \det \begin{bmatrix}
1        &   & 1        & 1       & 0 &  & 1       & 0 \\
\alpha_1 &   & \alpha_L & \beta_1 & 1 &  & \beta_M  & 1 \\
\alpha_1^2 & \cdots   & \alpha_L^2 & \beta_1^2 & 2 \beta_1 & \cdots  & \beta_M^2  &
2 \beta_M \\
\vdots & \ddots & \vdots & & \vdots & \ddots & \vdots \\
\alpha_1^{N-1} & \cdots & \alpha_L^{N-1} & \beta_1^{N-1} & (N-1)
\beta_1^{N-2} & \cdots & \beta_M^{N-1} & (N-1)\beta_M^{N-2}
\end{bmatrix} \nonumber \\
& \hspace{4cm}  = \prod_{j < k} ( \alpha_k -
\alpha_j ) \prod_{m < n} (\beta_m - \beta_n)^4 \prod_{\ell=1}^L
\prod_{m=1}^M (\alpha_{\ell} - \beta_m)^2. \label{eq:1}
\end{align}
We will denote the matrix on the left hand side of (\ref{eq:1}) by
$\mathbf V(\bs \alpha, \bs \beta)$ and its determinant by
$\Delta(\bs \alpha, \bs \beta)$.  We will later use the fact that the
monomials which appear in the 
definition of $\mathbf V(\bs \alpha, \bs \beta)$ can be replaced by any
family of monic polynomials $\mathbf p = (p_0, p_2, \ldots, p_{N-1})$
with $\deg p_n = n$ without changing the determinant.  We will write
the resulting matrix $\mathbf V^{\mathbf p}(\bs \alpha, \bs \beta)$,
and we note that $\Delta(\bs \alpha, \bs \beta) = \det \mathbf
V^{\mathbf p}(\bs \alpha, \bs \beta)$.   

It follows from (\ref{eq:2}) and (\ref{eq:2.5}) that 
\[
Z^{\nu_1, \nu_2}_{L,M} = \frac{1}{L! M!} \int_{\R^L} \int_{\R^M}
| \det \mathbf V^{\mathbf p}(\bs \alpha, \bs
\beta) |
d\nu_1^L(\alpha) \, d\nu_2^M(\beta).
\]

\subsubsection{Notation for minors}

Given a non-negative integer $L$, we denote the set $\{1,2, \ldots,
L\}$ by $\ul L$.  By convention, if $L = 0$, then $\ul L$ is the empty
set.  Given a function $\mf t: \ul L \nearrow \ul N$ we denote by $\mf
t'$ the unique function $\ul{N - L} \nearrow \ul N$ whose range is
disjoint from $\mf t$.  We denote by $\mf i$ the function $\ul L
\nearrow \ul N$ which is the identity on $\ul L$.  

We define $\sgn \mf t$ as follows:  Let $\mathbf e_1, \mathbf e_2,
\ldots, \mathbf e_N$ be any particular basis 
for $\R^N$.  We then specify that
\[
\mathbf e_{\mf t(1)} \wedge \cdots \wedge \mathbf e_{\mf t(L)} \wedge \mathbf
e_{\mf t'(1)} \wedge \cdots \wedge \mathbf e_{\mf t'(N-L)} = \sgn \mf
t \cdot \mathbf e_{1} \wedge \mathbf e_2 \wedge \cdots \wedge \mathbf e_N.
\]
That is $\sgn \mf t = (-1)^k$ where $k$ is the number of
transpositions necessary to put the set 
\[
\mf t(1), \cdots \mf t(L),\mf t'(1), \cdots, \mf t'(N-L)
\]
 into order.  Clearly, $\sgn \mf i = 1$.

Given an increasing function $\mf t: \ul L \nearrow \ul N$ and a
vector $\bs \alpha \in \R^N$, we define the vector $\bs \alpha_{\mf t}
\in \R^L$ by 
\[
\bs \alpha_{\mf t} = \left(\alpha_{\mf t(1)}, \alpha_{\mf t(2)}, \ldots,
\alpha_{\mf t(L)} \right)
\]
If $\mf u: \ul L \nearrow \ul N$ is another increasing function, and 
$\mathbf A = [a_{m,n}]$ an $N \times N$ matrix , we define $\mathbf A_{\mf t,
  \mf u}$ to be the $L \times L$ minor of $\mathbf A$ given by
\[
\mathbf A_{\mf t, \mf u} = \left[ a_{\mf t(j), \mf u(k)} \right]_{j,k=1}^L.
\]

\subsubsection{The Laplace expansion of the determinant}
Using this notation, the Laplace expansion of the determinant is given
by
\[
\det \mathbf A = \sum_{\mf t: \ul L \nearrow \ul N} \sgn \mf t \cdot
\det \mathbf A_{\mf t, \mf i} \cdot \det \mathbf A_{\mf t', \mf i'}.
\]
In particular, the Laplace expansion of the determinant of $\mathbf
V^{\mathbf p}(\bs \alpha, \bs \beta)$ is given by
\[
\mathbf V^{\mathbf p}(\bs \alpha, \bs \beta) = \sum_{\mf t: \ul L
  \nearrow \ul N} \sgn \mf t \cdot \det \mathbf V^{\mathbf p}_{\mf t, \mf i}(\bs
\alpha) \cdot \det \mathbf V_{\mf t', \mf i'}(\bs \beta),
\]
where the notation reflects the fact that minors of the form $\mathbf
V_{\mf t, \mf i}$ depend only on $\bs \alpha$ and minors of the form
$\mathbf V_{\mf t', \mf i'}$ only depend on $\bs \beta$.  

\subsubsection{The Total Partition Function}

Using the previous definitions, we may write 
\begin{align*}
&| \det \mathbf V^{\mathbf p}(\bs \alpha, \bs
\beta) | = \sum_{\mf t: \ul L
  \nearrow \ul N} \sgn \mf t \bigg\{
\prod_{j<k} \sgn(\alpha_k - \alpha_j) \bigg\} \det \mathbf V^{\mathbf
  p}_{\mf t, \mf i}(\bs \alpha) 
 \det \mathbf V^{\mathbf p}_{\mf t', \mf i'}(\bs \beta),
\end{align*}
and
\begin{align*}
Z^{\nu_1, \nu_2}_{L,M} =  \sum_{\mf t: \ul L \nearrow \ul N} \sgn \mf
t \, & \frac{1}{L!} \int_{\R^L} \bigg\{
\prod_{j<k} \sgn(\alpha_k - \alpha_j) \bigg\} \det \mathbf V^{\mathbf
  p}_{\mf t, \mf i}(\bs \alpha) \, d\nu_1^L(\bs \alpha)  \\
\times & \frac{1}{M!} \int_{\R^M}\det \mathbf V^{\mathbf p}_{\mf t',
  \mf i'}(\bs \beta) \, d\nu_2^M(\bs \beta).
\end{align*}

We define
\[
A_{\mf t} = \frac{1}{L!} \int_{\R^L} \bigg\{
\prod_{j<k} \sgn(\alpha_k - \alpha_j) \bigg\} \det \mathbf V^{\mathbf
  p}_{\mf t, \mf i}(\bs \alpha) \, d\nu_1^L(\bs \alpha),
\]
and
\[
B_{\mf t'} = \frac{1}{M!} \int_{\R^M} \det \mathbf V^{\mathbf p}_{\mf
  t', \mf i'}(\bs \beta) \, d\nu_2^M(\bs \beta),
\]
so that 
\begin{equation}
\label{eq:4}
Z^{\nu_1, \nu_2}_{L,M} =  \sum_{\mf t: \ul L \nearrow \ul N} \sgn \mf
t \cdot A_{\mf t} B_{\mf t'}.
\end{equation}

\subsubsection{Simplifying $A_{\mf t}$}
If $L$ is even,
\[
\prod_{j<k} \sgn(\alpha_k - \alpha_j) = \Pf \mathbf T(\bs \alpha),
\]
where $\mathbf T(\bs \alpha)$ is the $L \times L$ antisymmetric matrix given
by
\[
\mathbf T(\bs \alpha) = \left[ \sgn(\alpha_k - \alpha_j) \right]_{j,k=1}^L.
\]
A similar formula is available when $L$ is odd, but since $N$ and
therefore $L$ is even, we will not need that here.  

It follows that
\[
A_{\mf t} = \frac{1}{L!} \int_{\R^L}\Pf \mathbf T(\bs \alpha) \cdot
\det \mathbf V^{\mathbf p}_{\mf t, \mf i}(\bs \alpha) \, d\nu_1^L(\bs \alpha),
\]
and
\[
\det \mathbf V^{\mathbf p}_{\mf t, \mf i}(\bs \alpha) = \sum_{\sigma
  \in S_L} \sgn \sigma \prod_{\ell=1}^L p_{\mf
  t(\ell) - 1}(\alpha_{\sigma(\ell)}),
\]
where $S_L$ is the symmetric group of $L$ elements, so that,
\[
A_{\mf t} = \frac{1}{L!} \sum_{\sigma \in S_L} \int_{\R^L} \bigg\{
 p_{\mf t(\ell)-1}(\alpha_{\sigma(\ell)})
\bigg\} \sgn \sigma \cdot \Pf \mathbf T(\bs \alpha)  \,
d\nu_1^L(\bs \alpha).
\]
Now, $S_L$ acts on $\R^L$ by permuting coordinates---denote the action
of $\sigma$ by $\bs \alpha \mapsto \sigma \cdot \bs \alpha$.  It is
easy to verify that 
\[
\Pf \mathbf T(\sigma \cdot \bs \alpha) = \sgn \sigma \cdot \Pf \mathbf
T(\bs \alpha).
\]
Consequently,
\[
A_{\mf t} = \frac{1}{L!} \sum_{\sigma \in S_L} \int_{\R^L} \bigg\{
\prod_{\ell=1}^L p_{\mf t(\ell)-1}(\alpha_{\sigma(\ell)})
\bigg\} \Pf \mathbf T(\sigma \cdot \bs \alpha)  \,
d\nu_1^L(\bs \alpha),
\]
and by reindexing the integral by $\bs \alpha \mapsto \sigma^{-1} \cdot \bs
\alpha$, we find that
\begin{align*}
A_{\mf t} &= \frac{1}{L!} \sum_{\sigma \in S_L} \int_{\R^L} \bigg\{
\prod_{\ell=1}^L p_{\mf t(\ell)-1}(\alpha_{\ell})
\bigg\} \Pf \mathbf T(\bs \alpha)  \,
d\nu_1^L(\bs \alpha)= \int_{\R^L} \bigg\{\prod_{\ell=1}^L p_{\mf
  t(\ell)-1}(\alpha_{\ell}) \bigg\} \Pf \mathbf T(\bs \alpha)  \,
d\nu_1^L(\bs \alpha)  
\end{align*}

Next, we write $L = 2K$ and expand $\Pf \mathbf T(\bs \alpha)$ as a
sum over the symmetric group:
\[
\Pf \mathbf T(\bs \alpha) = \frac{1}{2^{K} K!} \sum_{\tau \in S_L}
\sgn \tau \prod_{k=1}^K \sgn( \alpha_{\tau(2k)} - \alpha_{\tau(2k-1)} ),
\]
so that
\begin{align*}
A_{\mf t} &= \frac{1}{2^K K!} \sum_{\tau \in S_L} \sgn \tau \int_{\R^L}
 \bigg\{ \prod_{k=1}^K p_{\mf
  t(2k-1)-1}(\alpha_{2k-1}) p_{\mf  t(2k)-1}(\alpha_{2k})
 \sgn( \alpha_{\tau(2k)} -
\alpha_{\tau(2k-1)} ) \bigg\} \, d\nu_1^L(\bs \alpha) \\
&= \frac{1}{2^K K!} \sum_{\tau \in S_L} \sgn \tau \\
& \hspace{1.2cm} \times \int_{\R^L}
\bigg\{ \prod_{k=1}^K  p_{\mf
  t \circ \tau(2k-1)-1}(\alpha_{\tau(2k-1)}) p_{\mf
  t \circ\tau(2k)-1}(\alpha_{\tau(2 k)})  \sgn( \alpha_{\tau(2k)} -
\alpha_{\tau(2k-1)} ) \bigg\} \, d\nu_1^L(\bs \alpha) \\
&= \frac{1}{2^K K!} \sum_{\tau \in S_L} \sgn \tau \prod_{k=1}^K
\int_{\R^2} p_{\mf t \circ \tau(2k-1)-1}(x) p_{\mf t \circ
  \tau(2k)-1}(y) \sgn(y-x) \, d\nu_1(x) \, d\nu_1(y) 
= \Pf \mathbf A^{\mathbf p}_{\mf t}.
\end{align*}

\subsubsection{Simplifying $B_{\mf t'}$} 
Turning to $B_{\mf t'}$, we first write
\[
\det \mathbf V_{\mf t', \mf i'}^{\mathbf p}(\bs \beta) = 
\sum_{\sigma \in S_{2M}} \sgn \sigma \prod_{m=1}^M p_{\mf t' \circ
  \sigma(2m-1)-1}(\beta_m) p'_{\mf t' \circ \sigma(2m)-1}(\beta_m),
\]
so that
\begin{align*}
B_{\mf t'} &= \frac{1}{M!} \sum_{\sigma \in S_{2M}} \sgn \sigma
\int_{\R^M} \bigg\{ \prod_{m=1}^M p_{\mf t' \circ
  \sigma(2m-1)-1}(\beta_m) p'_{\mf t' \circ \sigma(2m)-1}(\beta_m) \bigg\}
\, d\nu_2^M(\bs \beta) \\
&= \frac{1}{M!} \sum_{\sigma \in S_{2M}} \sgn \sigma \prod_{m=1}^M
\int_{\R} p_{\mf t' \circ \sigma(2m-1)-1}(\beta) p'_{\mf t'
  \circ \sigma(2m)-1}(\beta) \, d\nu_2(\beta).
\end{align*}
Next we set
\[
\Pi_{2M} = \{ \sigma \in S_{2M} : \sigma(2m-1) < \sigma(2m) \mbox{ for
} m=1,2,\ldots,M \},
\]
then 
\begin{align*}
B_{\mf t'} &= \frac{1}{M!} \sum_{\sigma \in \Pi_{2M}} \sgn \sigma
\prod_{m=1}^M \int_{\R} \left[ p_{\mf t' \circ
    \sigma(2m-1)-1}(\beta) p'_{\mf t' \circ \sigma(2m)}(\beta) - 
p_{\mf t' \circ \sigma(2m)}(\beta) p'_{\mf t' \circ
  \sigma(2m-1)-1}(\beta) \right]
\, d\nu_2(\beta) \\
&= \Pf \mathbf B_{\mf t'}^{\mathbf p}.
\end{align*}

\subsubsection{The Pfaffian formulation for the total partition
  function}

It follows from (\ref{eq:4}) that 
\[
Z^{\nu_1, \nu_2}_{L,M} = \sum_{\mf t: \ul L \nearrow \ul N} \sgn \mf t \cdot \Pf
\mathbf A^{\mathbf p}_{\mf t} \cdot \Pf \mathbf B^{\mathbf p}_{\mf t'},
\]
and therefore that,
\begin{align*}
Z^{\nu_1, \nu_2}(X) &= \sum_{(L,M)} X^L \sum_{\mf t:
  \ul L \nearrow \ul N} \sgn \mf t \cdot \Pf 
\mathbf A^{\mathbf p}_{\mf t} \cdot \Pf \mathbf B^{\mathbf p}_{\mf t'}
\\ 
&= \sum_{(L,M)} \sum_{\mf t:
  \ul L \nearrow \ul N} \sgn \mf t \cdot  X^L  \Pf 
\mathbf A^{\mathbf p}_{\mf t} \cdot \Pf \mathbf B^{\mathbf p}_{\mf t'}.
\end{align*}
Next, we set $\mathbf X = X \mathbf I$ where $\mathbf I$ is the $N
\times N$ identity matrix.  It is clear that, for each $\mf t: \ul L
\nearrow \ul N$, 
\[
(\mathbf X \mathbf A^{\mathbf p} \mathbf X^{\transpose})_{\mathbf t} =
\mathbf X_{\mathbf t} \mathbf A_{\mf t}^{\mathbf p} \mathbf X_{\mathbf
  t}^{\transpose},
\]
and 
\[
\Pf (\mathbf X_{\mathbf t} \mathbf A_{\mf t}^{\mathbf p} \mathbf X_{\mathbf
  t}^{\transpose}) = \det \mathbf X_{\mf t} \cdot \Pf \mathbf A_{\mf
  t}^{\mathbf p} = X^L \Pf \mathbf A_{\mf t}^{\mathbf p}.
\]
It follows that 
\begin{align*}
Z^{\nu_1, \nu_2}(X) &= \sum_{(L,M)} \sum_{\mf t:
  \ul L \nearrow \ul N} \sgn \mf t \cdot \Pf (\mathbf X \mathbf
A^{\mathbf p} \mathbf X^{\transpose})_{\mathbf t}  \cdot \Pf \mathbf
B^{\mathbf p}_{\mf t'} = \Pf( \mathbf X \mathbf A^{\mathbf p} \mathbf
X^{\transpose} + \mathbf B^{\mathbf p})
\end{align*}
where the last equation comes from the formula for the Pfaffian of a
sum of antisymmetric matrices \cite{MR1069389}.  Finally, since
$\mathbf X \mathbf A^{\mathbf p} \mathbf X^{\transpose} = X^2 \mathbf
A^{\mathbf p}$, we arrive at Theorem~\ref{thm:3}.  

\subsection{Proof of Theorem~\ref{thm:2}}

Given $x_1, \ldots, x_N, y_1, \ldots, y_N \in \R$ and indeterminants
$a_1, \ldots, a_N, b_1, \ldots, b_N$ we define the measures $\eta_1$
and $\eta_2$ on $\R$ given by 
\[
d\eta_1(\alpha) = \sum_{n=1}^N a_n d\delta(\alpha - x_n) \qquad
\mbox{and} \qquad 
d\eta_2(\beta) = \sum_{n=1}^N b_n d\delta(\beta - y_n),
\]
where $\delta$ is the probability measure with unit point mass at 0.
Using these measures, we will specialize the situation in
Section~\ref{sec:proof-theor-refthm:1} to 
the measures $\nu_1 = w (\mu + \eta_1)$ 
and $\nu_2 = w^2(\mu + \eta_2)$.  We will derive a Pfaffian form for
the correlation functions of the microcanonical ensemble with weight
$w$ by expanding both the integral and Pfaffian sides of (\ref{eq:19}) 
for this choice of $\nu_1$ and $\nu_2$ and equating coefficients of the
various products of the indeterminants.  

\subsubsection{Expanding the Integral Definition of $Z^{\nu_1, \nu_2}$} 

Starting with (\ref{eq:5}), and setting $Z^{\nu_1, \nu_2} = Z^{\nu_1,
  \nu_2}(1)$, we have
\begin{align*}
\frac{Z^{\nu_1, \nu_2}}{Z} &= \sum_{(L,M)} \frac{1}{L! M!}
\int_{\R^L} \int_{\R^M} \Omega_{L,M}(\bs \alpha, \bs \beta)
\, d(\mu + \eta_1)^L(\bs \alpha) \, d(\mu + \eta_2)^M(\bs \beta).
\end{align*}
It is easily verified that
\[
d(\mu + \eta_1)^L(\bs \alpha) = \sum_{\ell=0}^L \sum_{\mf u: \ul \ell
  \nearrow \ul L} \, d\eta_1^{\ell}(\bs \alpha_{\mf u}) \, d\mu^{L -
  \ell}(\bs \alpha_{\mf u'})
\]
and
\[
d(\mu + \eta_2)^M(\bs \beta) = \sum_{m=0}^M \sum_{\mf v: \ul m
  \nearrow \ul M} \, d\eta^m_2(\bs \beta_{\mf v}) \, d\mu^{M-m}(\bs
\beta_{\mf v'})
\]
so that,
\begin{align*}
  \frac{Z^{\nu_1, \nu_2}}{Z} &= \sum_{(L,M)}
  \sum_{\ell=0}^L \sum_{m=0}^M \sum_{\mf u: \ul \ell \nearrow \ul L}
  \sum_{\mf v: \ul
    m \nearrow \ul M} \frac{1}{L! M!}  \\
  &\hspace{2.5cm} \times \int_{\R^L} \int_{\R^M} \Omega_{L,M}(\bs
  \alpha, \bs \beta) d\eta_1^{\ell}(\bs \alpha_{\mf u}) \, d\mu^{L -
    \ell}(\bs \alpha_{\mf u'}) \, d\eta^m_2(\bs \beta_{\mf v}) \,
  d\mu^{M-m}(\bs \beta_{\mf v'}). 
\end{align*}
We may relabel the $\alpha$ and the $\beta$ in any manner we see fit.
In particular, we may replace each $\mf u: \ul \ell \nearrow \ul L$
and $\mf v: \ul m \nearrow \ul M$ by $\mf i: \ul \ell \nearrow \ul L$
and $\mf i: \ul m \nearrow \ul M$ respectively (the redundancy in the
notation should cause no confusion).  We compensate by a factor of ${L
\choose \ell}{ M \choose m}$ which counts the number of pairs $(\mf u,
\mf v)$ we are summing over.  That is,
\begin{align*}
\frac{Z^{\nu_1, \nu_2}}{Z} &= \sum_{(L,M)} \sum_{\ell=0}^L
  \sum_{m=0}^M \frac{1}{\ell! (L-\ell)! m! (M-m)!} \\
& \hspace{3cm}\times \int_{\R^L}
  \int_{\R^M} \Omega_{L,M}(\bs \alpha, \bs \beta)
  \, d\eta_1^{\ell}(\bs \alpha_{\mf i}) \, d\mu^{L - \ell}(\bs
  \alpha_{\mf i'}) \, d\eta^m_2(\bs \beta_{\mf i}) \, d\mu^{M-m}(\bs
  \beta_{\mf i'}).
\end{align*}
Now,
\[
d\eta_1^{\ell}(\bs \alpha_{\mf i}) = \prod_{j=1}^{\ell} \sum_{\mf u=1}^N
a_u \, d\delta(\alpha_j - x_u) \qquad \mbox{and} \qquad
d\eta_2^{\ell}(\bs \beta_{\mf i}) = \prod_{k=1}^{m} \sum_{v=1}^N
b_v \, d\delta(\beta_k - y_v),
\]
and exchanging the sum and product in each of these expressions,
\[
d\eta_1^{\ell}(\bs \alpha_{\mf i}) = \sum_{\mf u: \ul \ell \rightarrow
\ul N} \prod_{j=1}^{\ell} a_{\mf u(j)} \, d\delta(\alpha_j - x_{\mf
u(j)}) \quad \mbox{and} \quad 
d\eta_2^{\ell}(\bs \beta_{\mf i}) = \sum_{\mf v: \ul m \rightarrow
\ul N} \prod_{k=1}^{m} a_{\mf v(k)} \, d\delta(\beta_k - y_{\mf
v(k)}).
\]
Notice that the sums in the latter two expressions are over all (not
only increasing) functions from $\ul \ell$ and $\ul m$ into $\ul N$.  

So far, we have that
\begin{align*}
  \frac{Z^{\nu_1, \nu_2}}{Z} &= \sum_{(L,M)} \sum_{\ell=0}^L
  \sum_{m=0}^M \sum_{\mf u: \ul \ell \rightarrow \ul N} \sum_{\mf v:
    \ul m \rightarrow
    \ul N} \frac{1}{\ell! (L-\ell)! m! (M-m)!}   \int_{\R^L}
  \int_{\R^M} \Omega_{L,M}(\bs \alpha, \bs \beta) \\ 
  & \hspace{.5cm}\times \bigg\{ \prod_{j=1}^{\ell} a_{\mf u(j)} 
  d\delta(\alpha_j - x_{\mf u(j)}) \bigg\} d\mu^{L - \ell}(\bs \alpha_{\mf
    i'}) \bigg\{ \prod_{k=1}^{m} b_{\mf v(k)}  d\delta(\beta_k
  - y_{\mf v(k)}) \bigg\} d\mu^{M-m}(\bs \beta_{\mf i'}) \\
&= \sum_{(L,M)} \sum_{\ell=0}^L
  \sum_{m=0}^M \sum_{\mf u: \ul \ell \rightarrow \ul N} \sum_{\mf v:
    \ul m \rightarrow
    \ul N} \frac{1}{\ell! (L-\ell)! m! (M-m)!} \bigg\{ \prod_{j=1}^{\ell}
  a_{\mf u(j)} \prod_{k=1}^m b_{\mf v(k)} \bigg\} \\ 
  & \hspace{3.1cm}\times   \int_{\R^{L-\ell}} \int_{\R^{M-m}}
  \Omega_{L,M}(\mathbf x_{\mf u} \vee \bs \alpha, \mathbf y_{\mf v}
  \vee \bs \beta)  d\mu^{L - \ell}(\bs \alpha_{\mf
    i'})  d\mu^{M-m}(\bs \beta_{\mf i'}).
\end{align*}
Next we notice that, if $\mf u$ or $\mf v$ are not injective,
\[
  \Omega_{L,M}(\mathbf x_{\mf u} \vee \bs \alpha, \mathbf y_{\mf v}
  \vee \bs \beta) = 0,
\]
and we can therefore replace the sums over $\ul \ell \rightarrow \ul
N$ and $\ul m \rightarrow \ul N$ with sums over $\ul \ell \hookrightarrow
\ul N$ and $\ul m \hookrightarrow \ul N$.  In fact, since
$\Omega_{L,M}(\mathbf x_{\mf u} \vee \bs \alpha, \mathbf y_{\mf v}
\vee \bs \beta)$ is symmetric in the coordinates of $\mathbf x_{\mf
  u}$ and $\mathbf y_{\mf v}$, we may replace these sums with sums
over $\ul \ell \nearrow \ul N$ and $\ul m \nearrow \ul N$ so long as
we compensate by factors of $\ell!$ and $m!$.  Putting these
observations with the definition of $R_{\ell, m}$, we arrive at the
fact that
\begin{equation}
\label{eq:8}
  \frac{Z^{\nu_1, \nu_2}}{Z} = \sum_{(L,M)} \sum_{\ell=0}^L
  \sum_{m=0}^M \sum_{\mf u: \ul \ell \nearrow \ul N} \sum_{\mf v:
    \ul m \nearrow
    \ul N}  \bigg\{ \prod_{j=1}^{\ell}
  a_{\mf u(j)} \prod_{k=1}^m b_{\mf v(k)} \bigg\} R_{\ell, m}(\bs
  \alpha_{\mf u}; \bs \beta_{\mf v}). \\ 
\end{equation}
That is, $Z^{\nu_1, \nu_2}/Z$ is the generating function for the
correlation functions of our microcanonical ensemble.  

\subsubsection{Expanding the Pfaffian Formulation of $Z^{\nu_1, \nu_2}$}

It is easily computed that
\begin{align*}
& \la f | g \ra_1^{\nu_1} = \la f | g \ra_1 + 
2 \sum_{n=1}^N a_n \left[ \wt f(x_n) \epsilon_1 \wt g(x_n) - 
  \wt g(x_n) \epsilon_1 \wt f(x_n) \right] \\
& \hspace{5cm} - 2 \sum_{n=1}^N \sum_{m=1}^N a_n a_m \wt f(x_n) 
\wt g(x_m) \frac{\sgn(x_n - x_m)}2,
\end{align*}
and
\[
\la f | g \ra_4^{\nu_2} = \la f | g \ra_4 + \sum_{n=1}^N b_n
\left[  \wt f(y_n) \epsilon_2 \wt g (y_n) - \wt g(y_n) \epsilon_2 \wt
  f(y_n) \right]. 
\]
For convenience let us write
\[
\mathcal{E}_{1,1}(x,y) = \frac{1}{4} \sgn(y - x) \qquad \mbox{and} \qquad
\mathcal{E}_{1,2}(x,y) = \mathcal{E}_{2,1}(x,y) = \mathcal{E}_{2,2}(x,y) = 0,
\]
and define
\[
i(n) = \left\{
\begin{array}{lc}
1 & \quad 1 \leq n \leq N; \\
2 & \quad N < n \leq 2N,
\end{array},
\right.
\]
\[
c_n = \left\{ 
\begin{array}{lc}
2 a_n & \quad 1 \leq n \leq N; \\
b_{n-N} & \quad N < n \leq 2N,
\end{array}
\right. \qquad \mbox{and} \qquad
z_n = \left\{ 
\begin{array}{lc}
x_n & \quad 1 \leq n \leq N; \\
y_{n-N} & \quad N < n  \leq 2N.
\end{array}
\right.
\]
so that, 
\begin{align*}
& \la f | g \ra_1^{\nu_1} + \la f | g \ra_4^{\nu_2} = \la f | g \ra_1
+ \la f | g \ra_4  + 
\sum_{n=1}^{2N} c_n \left[ \wt f(z_n) \epsilon_{i(n)} \wt g(z_n) - 
  \wt g(z_n) \epsilon_{i(n)} \wt f(z_n) \right] \\
& \hspace{5cm} -  \sum_{n=1}^{2N} \sum_{m=1}^{2N} c_n c_m
\wt f(z_n) \wt g(z_m) \mathcal{E}_{i(m),i(n)}(z_m, z_n). 
\end{align*}

Defining $\mathbf A^{\mathbf p}$ and $\mathbf B^{\mathbf p}$ as in
Theorem~\ref{thm:1} and 
\[
\mathbf A^{\mathbf p, \nu_1} = \left[ \la p_j | p_k \ra_1^{\nu_1}
\right]_{j,k=1}^N, \qquad \mathbf B^{\mathbf p, \nu_1} = \left[ \la
  p_j | p_k \ra_1^{\nu_1} \right]_{j,k=1}^N, 
\]
we immediately see that 
\[
\mathbf A^{\mathbf p,\nu_1} + \mathbf B^{\mathbf p,\nu_2} =
\underbrace{\mathbf C^{\mathbf p}}_{= \mathbf A^{\mathbf p} + \mathbf
  B^{\mathbf p}} + \mathbf W^{\mathbf p} ,  
\]
where the $j,k$ entry of $\mathbf W^{\mathbf p}$ is given by
\begin{align*}
& \sum_{n=1}^{2N} c_n \left[ \wt p_j(z_n) \epsilon_{i(n)} \wt p_k(z_n)
  -  \wt p_k(z_n) \epsilon_{i(n)} \wt p_j(z_n) \right] \\
& \hspace{4cm} - \sum_{n=1}^{2N} \sum_{m=1}^{2N} c_n c_m
\wt p_j(z_n) \wt p_k(z_m) \mathcal{E}_{i(m),i(n)}(z_m, z_n).
\end{align*}

Next we define the $N \times 4N$ matrix $\mathbf X$ by
\begin{align*}
& \mathbf X = \bigg[ \sqrt{c_m} \wt p_j(z_m) \quad \sqrt{c_m}
  \epsilon_{i(m)} \wt p_j(z_m) \bigg] \qquad j=0,\ldots N-1; \quad
  m=1,\ldots,2N. 
\end{align*}
We also define the $4N \times 4N$ matrix $\mathbf{J}$ by 
\[
\mathbf J = \begin{bmatrix}
0 & 1 \\
-1 & 0 \\
   &   & \ddots \\
   &   &  & 0 & 1 \\
   &   &  & -1 & 0
\end{bmatrix},
\]
and the $4N \times 4N$ matrix $\mathbf Y$ by 
\[
\mathbf Y = -\mathbf J + \begin{bmatrix}
\sqrt{c_m c_n} \, \mathcal{E}_{i(m),i(n)}(z_m, z_n) & 0 \\
0 & 0 
\end{bmatrix}_{m,n=1}^{2N}.
\]
Finally, we set $\mathbf Z = (\mathbf C^{\mathbf p})^{-\transpose} =
\left[ \zeta_{j,k} \right]_{j,k=0}^{N-1}$.

A bit of matrix algebra reveals that
\[
\mathbf A^{\mathbf p,\nu_1} + \mathbf B^{\mathbf p, \nu_2} = \mathbf
Z^{-\transpose} + \mathbf{X Y X}^{\transpose},
\]
and therefore, 
\[
\frac{Z^{\nu_1, \nu_2}}{Z} = \frac{ \Pf(\mathbf
Z^{-\transpose} - \mathbf{X Y X}^{\transpose} ) }{\Pf (\mathbf Z^{-\transpose})} 
\]
This is useful, since by the Pfaffian Cauchy-Binet identity
\cite{rains-2000, borodin-2008},
\[
\frac{Z^{\nu_1, \nu_2}}{Z} = \frac{ \Pf(\mathbf
Y^{-\transpose} - \mathbf{X}^{\transpose} \mathbf{ Z X} ) }{\Pf (\mathbf
Y^{-\transpose})}  
\]

A bit more matrix algebra reveals that
\begin{align*}
& \mathbf{X}^{\transpose} \mathbf{ Z X} =\begin{bmatrix}
\sqrt{c_n c_m} \sum\limits_{j,k=0}^{N-1} \wt p_j(z_m)
\zeta_{j,k} \wt p_k(z_n) & \sqrt{c_n c_m} \sum\limits_{j,k=0}^{N-1}
\! \wt p_j(z_m) \zeta_{j,k} \epsilon_{i(n)} \wt p_k(z_n) \\
\sqrt{c_n c_m} \sum\limits_{j,k=0}^{N-1} \! \epsilon_{i(m)} \wt p_j(z_m)
\zeta_{j,k} \wt p_k(z_n) & \sqrt{c_n c_m} \sum\limits_{j,k=0}^{N-1}
\! \epsilon_{i(m)} \wt p_j(z_m) \zeta_{j,k} \epsilon_{i(n)} \wt p_k(z_n)
\end{bmatrix}_{m,n=1}^{2N},
\end{align*}
And since 
\[
\mathbf Y^{-\transpose} = -\mathbf J - \begin{bmatrix} 0 & 0 \\
0 & \sqrt{c_m c_n} \, \mathcal{E}_{i(m),i(n)}(z_m, z_n)
\end{bmatrix}_{m,n=1}^{2N},
\]
we have that $\Pf( \mathbf Y^{-\transpose}) = (-1)^N$.  This implies
that
\[
\frac{Z^{\nu_1, \nu_2}}{Z} = \Pf\left(\mathbf X^{\transpose} \mathbf{ZX} -
\mathbf Y^{-\transpose}\right)
\]

Looking at the entries of $\mathbf X^{\transpose} \mathbf{ZX}$ and
using the definitions of $\varkappa_N$, $\epsilon_1$ and $\epsilon_2$,
we see 
\begin{align*}
\sum_{j,k=0}^{N-1} \wt p_j(z_m) \zeta_{j,k}
\wt p_k(z_n) &= \varkappa_N(z_m, z_n) \\ 
 \sum_{j,k=0}^{N-1} \wt p_j(z_m) \zeta_{j,k}
\epsilon_{i(n)} \wt p_k(z_n) &= \varkappa_N \epsilon_{i(n)} (z_m, z_n)
\\
 \sum_{j,k=0}^{N-1} \epsilon_{i(m)} \wt p_j(z_m) \zeta_{j,k}
\wt p_k(z_n) &= \epsilon_{i(m)} \varkappa_N  (z_m, z_n) \\ 
\sum_{j,k=0}^{N-1} \epsilon_{i(m)} \wt p_j(z_m) \zeta_{j,k}
\epsilon_{i(n)} \wt p_k(z_n) &= \epsilon_{i(m)} \varkappa_N(z_m, z_n)
\epsilon_{i(n)}.
\end{align*}
It follows that 

\[
\frac{Z^{\nu_1, \nu_2}}{Z} = (-1)^N \Pf( - \mathbf J - \mathbf K) =
\Pf(\mathbf J + \mathbf K),
\]
where 
\[
\mathbf K = \bigg[ \sqrt{c_m c_n} K^{i(m), i(n)}_N(z_m, z_n) \bigg]_{m,n=1}^{2N}.
\]

Using these definitions
and the formula for the Pfaffian of the sum $\mathbf J +
\mathbf K$, \cite{MR1069389}, we find that 
\[
\frac{Z^{\nu_1, \nu_2}}{Z} = \sum_{n=1}^{2N} \sum_{\mf t: \ul{n}
  \nearrow \ul{2N}} \Pf \mathbf K_{\mf t},
\]
where $\mathbf K_{\mf t}$ is the $2n \times 2n$
antisymmetric matrix given by
\[
\mathbf K_{\mf t} = \bigg[ \sqrt{c_{\mf t(j)} c_{\mf t(k)}}
  K^{i \circ \mf t(j), i \circ \mf t(k)}_N(z_{\mf t(j)},
  z_{\mf t(k)}) \bigg]_{j,k=1}^n. 
\]

For each $\mf t: \ul n \nearrow \ul{2N}$ there exists non-negative 
integers $\ell$ and $m$ such that $\ell + m = n$ and functions $\mf u:
\ul \ell \nearrow \ul N$ and $\mf v: \ul m \nearrow \ul N$ given by
\[
\mf u(j) = \mf t(j) \qquad \mbox{for all $j$ such that $\mf t(j) \leq N$}
\]
and
\[
\mf v(j) = \mf t(j) - N \qquad \mbox{for all $j$ such that $N < \mf t(j) \leq 2N$}.
\]
In this situation we write $\mf t = \mf u \vee \mf v$.  (We allow for
the possibility that $\mf u$ or $\mf v$ is the ``empty'' function, in
which case $\mf u \vee \mf v = \mf v$ or $\mf u \vee \mf v = \mf u$
respectively.)  It  follows that we may write
\[
\frac{Z^{\nu_1, \nu_2}}{Z} = \sum_{\ell = 0}^N \sum_{m=0}^N \sum_{\mf
  u: \ul \ell \nearrow \ul N} \sum_{\mf v: \ul m \nearrow \ul N} \Pf
\mathbf K_{\mf u \vee \mf v}
\]
Finally,
\[
\Pf K_{\mf u \vee \mf v} = \bigg\{ 2^{\ell} \prod_{j=1}^{\ell} a_{\mf u(j)}
\prod_{k=1}^m b_{\mf v(k)} \bigg\} \Pf 
\begin{bmatrix}
K^{1,1}_N(x_{\mf u(j)}, x_{\mf u(j')}) &
K^{1,2}_N(x_{\mf u(j)}, y_{\mf v(k')}) \\ 
K^{2,1}_N(y_{\mf v(k)}, x_{\mf u(j')}) & 
K^{2,2}_N(y_{\mf v(k)}, y_{\mf v(k')})
\end{bmatrix}_{j,j',k,k'=1}^{\ell,\ell,m,m}
\]
and thus,
\begin{align}
&\frac{Z^{\nu_1, \nu_2}}{Z} = \sum_{\ell = 0}^N \sum_{m=0}^N \sum_{\mf
  u: \ul \ell \nearrow \ul N} \sum_{\mf v: \ul m \nearrow \ul N} 
\bigg\{ \prod_{j=1}^{\ell} a_{\mf u(j)}
\prod_{k=1}^m b_{\mf v(k)} \bigg\} \label{eq:7} \\
& \hspace{5cm} \times 2^{\ell} \Pf 
\begin{bmatrix}
K^{1,1}_N(x_{\mf u(j)}, x_{\mf u(j')}) &
K^{1,2}_N(x_{\mf u(j)}, y_{\mf v(k')}) \\ 
K^{2,1}_N(y_{\mf v(k)}, x_{\mf u(j')}) & 
K^{2,2}_N(y_{\mf v(k)}, y_{\mf v(k')})
\end{bmatrix}_{j,j',k,k'=1}^{\ell,\ell,m,m} \nonumber
\end{align}

Comparing the coefficient of $a_1 a_2 \cdots a_{\ell} b_1 b_2 \cdots
b_m$ in this expression with that in (\ref{eq:8}) we find that
\[
R_{\ell,m}(\mathbf x; \mathbf y) = 2^{\ell} \Pf 
\begin{bmatrix}
K^{1,1}_N(x_j, x_{j'}) & K^{1,2}_N(x_j, y_{k}) \\
K^{2,1}_N(y_{k}, x_{j'}) & K^{2,2}_N(y_{k}, y_{k'})
\end{bmatrix}_{j,j',k,k'=1}^{\ell,\ell,m,m},
\]
as desired.

\subsection{Proof of Theorem~\ref{thm:5}}

Let $H_n$ be the standard Hermite polynomial. It is known (cf. \cite{MR1762659})
\begin{equation} \label{Hn-ipd1}
  \la H_{2m}, H_{2n+1} - 4 n H_{2n-1} \ra_1 = 4 h_{2n} \delta_{m,n}, 
\end{equation}
where 
$$
    h_n: = \int_{\RR} H_n^2(x) e^{-x^2} dx = \sqrt{\pi} 2^n n!.
$$
Using the fact that $H_k'(x) = 2k H_{k-1}(x)$, it follows readily that 
\begin{align} \label{Hn-ipd4}
 \la H_{2m}, H_{2n+1}  \ra_4  & = \int_{\RR} \left( H_{2m}(x) H_{2n+1}'(x) - H_{2m}'(x) H_{2n+1}(x) \right) dx \\
    & = 2(2 n+1) h_n \delta_{m,n} - 4m h_{2m-1} \delta_{m,n+1} \notag \\    
    & = h_{2n+1} \delta_{m,n} - h_{2n+1} \delta_{m,n+1}. \notag
\end{align}
We look for skew orthogonal polynomials in the form of  
$$
    P_{2m}(x) = \sum_{k=0}^m a_k  H_{2k}(x), \qquad a_0 =1, 
$$
and determine the coefficients $a_k$, $1 \le k \le m$, by the orthogonal conditions
$$
    \la P_{2m}, H_{2k+1} -  4 k H_{2k-1}\ra = 0, \quad k =0, 1, \ldots, m-1.
$$
This gives, by \eqref{Hn-ipd1} and \eqref{Hn-ipd4}, the following equations on $a_k$,
$$
 - h_{2k} a_{k-1} + \left[  h_{2k+1} + 4(k+ X^2) h_{2k}\right] a_k -  h_{2k+2} a_{k+1} =0, \quad k =0,1,\ldots,m-1,
$$
where we define $a_{-1} =0$. Rescale by setting $\wh a_k : = \frac{k!}{(2k)!} a_k$, the above 
equations become 
$$
  - (k+1) \wh a_{k+1} + (2 k+\tfrac12 -X^2) \wh a_k - (k+\tfrac12) \wh a_{k-1}, \quad k = 0,1,\ldots, m-1, 
$$
which can be used to determine $\wh a_k$ recursively, starting from $\wh a_{-1} =0$ and $\wh a_0 =1$. It 
turns out, however, that this is precisely the three-term recurrence relation for $L_k^{-\frac12}(-X^2)$. 
Consequently, $\wh a_k = L_k^{-\frac12}(-X^2)$. Hence, we conclude 
\begin{equation} \label{skewOP2}
   P_{2m} (x) = \sum_{k=0}^m a_k H_{2k}(x), \quad a_k = \frac{k!} {(2k)!} L^{-\frac12} (-X^2). 
\end{equation}
The Hermite polynomials are related to the Laguerre polynomials by
\cite[p. 106]{Sz} 
\begin{equation} \label{Hermite-Lagurre}
     H_{2k}(x) = (-1)^k 2^{2k} k! L_k^{-1/2}(x^2), \quad H_{2k+1}(x) = (-1)^k 2^{2k+1} k! x L_k^{-1/2}(x^2).
\end{equation}
Using this relation and the facts that 
\begin{equation} \label{factorial2n}
  L_k^{ \alpha}(0) = \binom{k+ \alpha}{k} \quad \hbox{and} \quad (2k)! = \frac{2^{2k}}{\sqrt{\pi}}
          \Gamma(k+\tfrac12) \Gamma(k+1),
\end{equation}
we see that \eqref{skewOP2} becomes \eqref{skewOP}. Since $P_{2m}$ is determined by \eqref{Hn-ipd1} and
\eqref{Hn-ipd4}, the same process shows that 
\begin{equation} \label{skewOPodd}
   P_{2m+1} (x) = \sum_{k=0}^m a_k (H_{2k+1}(x) - 4k H_{2k}(x)), \quad a_k = \frac{k!} {(2k)!} L^{-\frac12} (-X^2). 
\end{equation}
Using the fact that $H_{2j} '(x) = 4 j H_{2j-1}(x)$ and $H_{2j+1}(x) = 2x H_{2j} (x) - H_{2j}'(x)$, we see that 
$$
  P_{2m+1} (x) = \sum_{k=0}^m a_k (2 x H_{2k}(x) - 2 H_{2k}'(x)) = 2x P_{2m}(x) - 2 P_{2m}'(x). 
$$

It remains to compute $\la P_{2m}, P_{2m+1} \ra$. By \eqref{Hn-ipd1} and \eqref{Hn-ipd4}, we have
\begin{align*}
   \la P_{2m}, P_{2m+1} \ra  & = a_m \la P_{2m}, H_{2m+1} - 4m H_{2m-1} \ra  \\
         & = a_m [- h_{2m} a_{m-1} + (h_{2m+1} + 4(m+ X^2) h_{2m})a_{2m} ] \\ 
         & = a_m \sqrt{\pi} 2^{m+1} m! \left[ - (2m-1) L_{m-1}^{-\frac12}(- X^2)
             +  (4m + 2 X^2 +1) L_m^{-\frac12}(- X^2) \right] \\
         & = a_m \sqrt{\pi} 2^{m+1} m!  2(m+1) L_{m+1}^{-\frac12}(- X^2),
\end{align*}
where the last step follows from the three-term relation of the Laguerre polynomials, from which 
\eqref{skewNorm} follows readily. 

Finally, we turn to the proof of (\ref{skewOPodd2}).  Using
$\frac{d}{dx} L_k^\a(x) = - L_{k+1}^{\a+1}(x)$ and
$L_k^\a(x) = L_k^{\a+1}(x) - L_{k-1}^\a(x)$ (\cite[p. 102]{Sz}),
\eqref{skewOP} gives
\begin{align*}
& P_{2m+1}(x)  =  2 x \sum_{k=0}^{m} (-1)^k\wh L_k^{-\frac12}(- X^2) \left[L_k^{\frac12}(x^2) 
    + L_{k-1}^{\frac12}(x^2)\right]  \\
   & =  2 x \left[ \sum_{k=0}^{m-1} (-1)^k L_k^{\frac12}(x^2) \left(
         \wh L_k^{-\frac12}(- X^2) -  \wh L_{k+1}^{-\frac12}(- X^2)\right) 
      + (-1)^m \wh L_m^{-\frac12}(- X^2) L_{m}^{\frac12}(x^2) \right],  
\end{align*}
where $\wh L_k^\a(x) : = L_k^\a(x)/ L_k^\a(0)$, the formula \eqref{skewOPodd2} then follows from 
$$
     x  \wh L_k^\a(x) = - (\a+1) \left(\wh L_k^\a(x) - \wh L_{k-1}^\a(x)\right),  
$$
which is a rescaling of the identity $x L_k^{\a+1}(x) = - (k+1) L_{k+1}^\a(x) + (k+1+\a) L_k^\a(x)$ (\cite[p. 102]{Sz}).

\bibliography{bibliography}

\noindent\rule{4cm}{.5pt}
\vspace{.25cm}

\noindent {\sc \small Brian Rider}\\
{\small Department of Mathematics, University of Colorado, Boulder CO 80309} \\
email: {\tt brian.rider@colorado.edu}

\vspace{.25cm}

\noindent {\sc \small Christopher D.~Sinclair}\\
{\small Department of Mathematics, University of Oregon, Eugene OR 97403} \\
email: {\tt csinclai@uoregon.edu}

\vspace{.25cm}

\noindent {\sc \small Yuan Xu}\\
{\small Department of Mathematics, University of Oregon, Eugene OR 97403} \\
email: {\tt yuan@uoregon.edu}

\end{document}